\documentclass[letterpaper, 10 pt, conference]{ieeeconf}  % Comment this line out if you need a4paper

\IEEEoverridecommandlockouts                              % This command is only needed if 
\newif\ifarxiv

\newif\ifimages
\imagestrue
\arxivtrue
\usepackage{makeidx}         % allows index generation
\usepackage{graphicx}        % standard LaTeX graphics tool
\ifarxiv
\fi

\makeindex             % used for the subject index
\usepackage[bookmarks=true]{hyperref}
\usepackage{mathtools}

\usepackage{enumerate}
\usepackage{epstopdf}
\usepackage{subfigure}
\usepackage{cleveref}

\usepackage{tablefootnote}
\usepackage{amsmath,amssymb}%,amsthm}

\usepackage[usenames,dvipsnames,svgnames]{xcolor}

\graphicspath{{./fig/}}
\usepackage{color}
\usepackage{amsmath}
\usepackage{amssymb}
\usepackage{graphicx}
\usepackage{comment,xspace}
\usepackage{hyperref}
\usepackage{fancybox}

\newtheorem{theorem}{Theorem}[section]
\newtheorem{definition}[theorem]{Definition}

\newtheorem{remark}[theorem]{Remark}
\newtheorem{proposition}[theorem]{Proposition}
\newtheorem{alg}{Algorithm}

\renewcommand{\baselinestretch}{0.948}

\title{\LARGE \bf
Model Predictive Control of Autonomous Mobility-on-Demand Systems
}

\author{Rick Zhang, Federico Rossi, and Marco Pavone% <-this % stops a space
\thanks{This research was supported by National Science Foundation under CAREER Award CMMI-1454737 and by the Dr. Cleve B. Moler Stanford Graduate Fellowship.}% <-this % stops a space
\thanks{Rick Zhang, Federico Rossi, and Marco Pavone are with the Department of Aeronautics \& Astronautics, Stanford University, Stanford, CA 94305 {\tt\small \{rickz, frossi2, pavone\}@stanford.edu}}%
}

\begin{document}

\maketitle 
\thispagestyle{empty}
\pagestyle{empty}

\begin{abstract}
In this paper we present a model predictive control (MPC) approach to optimize vehicle scheduling and routing in an autonomous mobility-on-demand (AMoD) system. In AMoD systems, robotic, self-driving vehicles transport customers within an urban environment and are coordinated to optimize service throughout the entire network. Specifically, we first propose a novel discrete-time model of an AMoD system and we show that this  formulation allows the easy integration of a number of real-world constraints, e.g., electric vehicle charging constraints. Second, leveraging our model, we design a model predictive control algorithm for the optimal coordination of an AMoD system and prove its stability in the sense of Lyapunov. At each optimization step, the  vehicle scheduling and routing problem is solved  as a mixed integer linear program (MILP) where the decision variables are binary variables representing whether a vehicle will 1) wait at a station, 2) service a customer, or 3) rebalance to another station.  Finally, by using real-world data, we show that the MPC algorithm can be run in real-time for moderately-sized systems and outperforms previous control strategies for AMoD systems.
\end{abstract}

\section{Introduction}
Current mobility trends, anchored by the continued growth of privately owned automobiles, are creating high levels of air pollution and traffic congestion, especially in densely populated cities with limited space for road infrastructure and parking. Urban transportation in the US account for over half of the total oil consumption \cite{EIA:13} while producing 20\% of total carbon dioxide emissions \cite{UN:13}. With world urban population projected to increase by 2.5 billion by 2050 \cite{UN:14}, current urban transportation trends are widely viewed as unsustainable for the future. 

The eventual solution to this problem will likely involve the convergence of several key emerging technologies. First, one-way carsharing has emerged as a promising solution to increase vehicle utilization and promote sustainable urban land use, and the rise in mobile technology has enabled on-demand taxi services like Uber \cite{Uber:15}.
% and Lyft \cite{Lyft:15}. 
Second, electric vehicle technology has the potential to drastically reduce emissions and dependence on oil, and promote the generation of renewable energy. Finally, the advancement in autonomous driving technology promises to further increase convenience (through on-demand service), safety, and mobility for people unable or unwilling to drive. These emerging pieces lead to a transformational technology known as autonomous mobility-on-demand (AMoD) \cite{MP-SLS-EF-DR:12,WJM-CEBB-LDB:10,RZ-MP:15a}, whereby shared driverless electric cars provide personal on-demand transportation for customers. Its many potential benefits have led a number of companies to aggressively pursue AMoD technology \cite{Google:14,Ligier:14}. However, the optimal coordination of these robotic electric vehicles in a transportation network remains a challenge. 

\emph{Statement of contributions:} The objective of this paper is to design a model predictive control (MPC) approach to optimize vehicle scheduling and routing in an AMoD system. Model predictive control (also known as receding horizon control) is a control technique whereby an open-loop optimization problem is solved at each time step to yield a sequence of control actions up to a fixed horizon, and the first control action is executed. Due to its iterative nature, MPC can achieve closed-loop performance, is robust to model errors, and is well suited for complex, constrained systems. Previous work on mobility-on-demand (MoD) and AMoD systems has focused on design and fleet sizing, where the steady state of the system is characterized using a fluidic model \cite{MP-SLS-EF-DR:12,SLS-MP-MS-EF-DR:13}, a queueing network model \cite{RZ-MP:15a,KS-KT-RZ-EF-DM-MP:14,RZ-MP:15} or a Markov model \cite{MV-JA-DR:12}. Real-time control strategies devised in these works are heuristics based on the steady-state model and are exclusively concerned with the ``rebalancing'' problem (where the robotic vehicles redistribute themselves to align with asymmetric customer demand). In particular, these models do not allow easy integration of real-world constraints such as electric vehicle charging and parking capacities, which limits their practical application. Control algorithms for AMoD are also related to dispatch algorithms for taxis \cite{KTS-NHD-DHL:10,AA-SA-IR:09}, but these algorithms typically cannot enforce a vehicle distribution since taxi drivers make the final decisions. 

The key feature of our approach is that it is amenable to real-time optimization and receding horizon control while having the flexibility to account for many real-world phenomena such as battery charging constraints, customer priorities, and parking space limitations. Specifically, the contribution of this paper is threefold. First, we propose a novel discrete-time model of an AMoD system and we show that this formulation allows the easy integration of a number of real-world constraints, with a special focus on electric vehicle charging constraints. Second, leveraging our model, we design a model predictive control algorithm for the optimal coordination of an AMoD system and prove its stability in the sense of Lyapunov. Finally,  by using real-world data, we show that the MPC algorithm can be run in real-time for moderately-sized systems and compare its performance to four other AMoD control algorithms and taxi dispatch algorithms in the literature. We show that the MPC algorithm not only outperforms other algorithms in terms of customer wait times, but can also be used as an optimal performance benchmark to evaluate other dispatch algorithms. %Our approach can be considered the ``practical" complement to the theoretical models in \cite{MP-SLS-EF-DR:12,SLS-MP-MS-EF-DR:13,KS-KT-RZ-EF-DM-MP:14,RZ-MP:15a,RZ-MP:15}.

Our work draws inspiration from time-space network models such as \cite{NT-TM-LS:06} for optimizing bus routes and \cite{DJ-GHAC-CB:14,AGHK-RLC-QM-CHF:09} for vehicle redistribution policies in carsharing systems. Receding horizon (or model predictive) control techniques have been used in the context of transportation systems \cite{SL-BDS-YX-HH:12,RRN-BDS-JH:08}. The key difference between our approach and these works, besides the modeling differences, is that we provide a rigorous proof of stability within an MPC framework. In this perspective, our approach is related to the one used in \cite{CD-FB-DO-ea:12} for capacity maximization in battery networks and \cite{EF-LM-TP-ea:08,LM:05} for cooperative multi-agent systems.

\emph{Organization:}  The remainder of this paper is organized as follows: In Section \ref{sec:model} we present our AMoD model and discuss the inclusion of operational constraints, in particular battery charging. In Section \ref{sec:problem} we formulate the problem of regulating an AMoD system. In Section \ref{sec:mpc} we present two MPC algorithms to control the AMoD system and prove their stability. Simulation studies are presented in Section \ref{sec:results} to assess the performance of the MPC algorithms and characterize the effect of charging constraints on system throughput. In Section \ref{sec:extensions} we discuss the inclusion of additional operational constraints, e.g., customers' priorities. Lastly, in Section \ref{sec:conc} we draw our conclusions and present directions for future research. 

\section{Model} \label{sec:model}
In this section, we first introduce a linear discrete-time model of an AMoD system with similar assumptions as those proposed in the fluidic model in \cite{MP-SLS-EF-DR:12}. We then expand the model by introducing charging constraints associated with using electric vehicles. Such constraints are both of practical and theoretical interests. In particular, as we will see, with the addition of charging constraints (which are piecewise-linear), the system is no longer strictly a linear system. However, optimization in the form of a mixed-integer linear program (MILP) can still be performed with these additional constraints.  The inclusion of additional operational constraints is discussed in Section \ref{sec:extensions}.
%\mpmargin{}{As discussed, let's present charring constraints as an example of possible operational constraints.}

\subsection{Linear AMoD model} \label{sec:linear}
We consider a discrete-time system with $N$ stations and $m$ single-occupancy vehicles. Let $\mathcal{N}$ represent the set of stations, $|\mathcal{N}| = N$, and let $\mathcal{V}$ represent the set of vehicles, where $|\mathcal{V}| = m$. 
At each time step, customers arrive at each station and wait for vehicles to transport them to their desired destinations. In our model, customers may not be serviced on a first-come-first-serve basis, as the system determines the best ordering for serving the customers. While this idea may at first seem to be at odds with the notion of ``fairness," it is a standard strategy for ride-sharing services like SuperShuttle \cite{SuperShuttle:15}, where it is important to cluster customers traveling in the same direction. This strategy is also more natural in an equivalent system where ``stations'' are geographical regions (rather than physical infrastructure) and customers request transportation via a mobile app. Furthermore, we will show in Section \ref{sec:extensions} that customer priority can be easily integrated into the model.

% \rzmargin{We also consider an optimization horizon of $T$ time steps, $t \in [t_0, t_0+T]$. There are potentially many ways to model the system, and the complexity of the system makes it difficult to pinpoint the (state) variables that need to be propagated from one time period to the next. We therefore begin with the ``control'' variables of the system.}{move this somewhere else}

Before defining the system states, we first define the ``control'' variables of the AMoD system. A control decision is made at each time step for each vehicle parked at a station. The two possible actions each vehicle can take are 1) transport a customer from one station to another, and 2) rebalance the system by driving itself from one station to another (this is a key advantage of robotic vehicles). We can encode these actions using binary variables. Let $v^k_{ij}(t) = 1$ if vehicle $k$ is transporting a customer from station $i$ to station $j$ beginning at time $t$, and arriving at station $j$ at time $t+t_{ij}$. The travel time $t_{ij}$ is assumed to be deterministic and known. Similarly, let $w^k_{ij}(t) = 1$ if vehicle $k$ is rebalancing from station $i$ to station $j$ beginning at time $t$ and arriving at time $t+t_{ij}$. 

Denote by $d_{ij}(t)$ the number of customers waiting at station $i$ at the start of time period $t$ whose destination is station $j$. Denote by $c_{ij}(t)$ the number of customers that arrive at station $i$ at time $t$ heading to station $j$. The dynamics of $d_{ij}(t)$ are propagated as follows:

\vspace{-1mm}
{\small\begin{equation}
d_{ij}(t+1) = d_{ij}(t) + c_{ij}(t) - \sum_{k \in \mathcal{V}} v^k_{ij}(t),
\label{eq:dij}
\end{equation}}
\vspace{-2mm}

where the last term represents the number of passenger-carrying vehicles leaving station $i$ at time $t$. Note that $d_{ij} \geq 0$ for all $i, j \in \mathcal{N}$ and for all time. This means that if $d_{ij}(t) = 0$ and $c_{ij}(t) = 0$, then $v^k_{ij}(t) = 0$ for all $k \in \mathcal{V}$. The number of waiting customers plays a significant role in characterizing the performance of the system, hence $d_{ij}(t)$ is modeled as a state variable.

When a vehicle is on the road, it is necessary to keep track of how long it will be traveling before it reaches its destination. We represent this by the binary variables $^{T_i}p^k_{i}(t) \in \{0,1\}$, where $k \in \mathcal{V}$ and $T_i \in \{0, \max_j\{t_{ji}\}-1\}$ is the number of time steps remaining until the vehicle reaches station $i$, $i \in \mathcal{N}$. Suppose vehicle $k$ leaves station $j$ destined for station $i$ at time $t$; then $^{t_{ji}-1}p^k_i(t+1)$ is set to one at the next time step to indicate that the vehicle is $t_{ji}-1$ time steps from station $i$. In the subsequent time step, $^{t_{ji}-2}p^k_i(t+2)$ is set to one to indicate the progress of the vehicle along its path. Eventually $^{0}p^k_i(t+t_{ji})$ is set to one to signal that the vehicle has arrived at station $i$. The propagation of $^{T_i}p^k_i(t)$ is formally defined as follows:

\vspace{-3.5mm}
{\small
\begin{align}
&^{T_i}p^k_i(t+1) \notag \\
&=\begin{cases}
^{T_i+1}p^k_i(t) + \sum\limits_{j: t_{ji}-1=T_i} (v^k_{ji}(t) + w^k_{ji}(t)) & \text{if } T_i < T_{\text{max},i} \\
\sum\limits_{j: t_{ji}-1=T_{\text{max},i}} (v^k_{ji}(t) + w^k_{ji}(t)) & \text{if } T_i = T_{\text{max},i},
\end{cases}
\label{eq:tpki}
\end{align}
}

\vspace{-3mm}
where $T_{\text{max},i} = \text{max}_j t_{ji}-1$. Note that since each vehicle can only be in one place at one time, the $^{T_i}p^k_i$'s are subject to the constraints

\vspace{-3.5mm}
{\small\begin{align}
\sum_{i \in \mathcal{N}} \sum_{T_i}\ ^{T_i}p^k_i(t) \leq 1.
\label{eq:pconstraint}
\end{align}
}
Constraint \ref{eq:pconstraint} will not be enforced explicitly, and instead will result from constraints on the control variables. The dimension of $^{T_i}p^k_i$ depends on the travel time from each station to every other station. For each vehicle and each station $i$ there are $T_{\text{max},i}$ such variables, so the total dimension of $^{T_i}p^k_i$ is $\mathcal{D} = |\mathcal{V}| \sum_{i \in \mathcal{N}} T_{\text{max},i}$.

%\rzmargin{
%An important constraint on this formulation is ensuring that the binary decision variables $u^k_i(t)$, $v^k_{ij}(t)$, and $w^k_{ij}(t)$ remain zero when vehicle $k$ is traveling. To address this, we introduce a binary state variable $b^k(t)$ to denote whether vehicle $k$ is ``busy'' (i.e. on the road). $b^k(t)$ is governed by the following dynamics
%\begin{align}
%b^k(t) = b^k(t-1) &+ \sum_{i \in \mathcal{N}} \sum_{j \in \mathcal{N}} (v^k_{ij}(t) + w^k_{ij}(t)) \notag \\
%&- \sum_{i \in \mathcal{N}} \sum_{j \in \mathcal{N}} (v^k_{ji}(t-t_{ji}) + w^k_{ji}(t-t_{ji})).
%\end{align}}{replaced by p's}
Finally, let $u^k_i(t)$ be the state variable associated with waiting at a station, that is $u^k_i(t) = 1$ if vehicle $k$ waited at station $i$ from time $t-1$ to time $t$. The dynamics of $u^k_i(t)$ are modeled as follows
\begin{align}
u^k_i(t+1) = u^k_i(t) + ^0p^k_i(t) - \sum_{j \in \mathcal{N}} (v^k_{ij}(t) + w^k_{ij}(t)).
\vspace{-1.5em}
\label{eq:uit}
\end{align}
Equation \eqref{eq:uit} ensures that 1) a vehicle can only perform an action (via $v^k_{ij}$ or $w^k_{ij}$) if it is at a station, and 2) if a vehicle does not perform an action, it waits at a station. In other words, each vehicle must complete a task before starting another. However, a vehicle cannot perform an action if it is already on the road, thus a constraint is needed between $u^k_i(t)$ and $^{T_i}p^k_i(t)$ which ensures that a vehicle is either waiting at a station or traveling. This is formalized as

\vspace{-2mm}
{\small
\begin{align}
\sum_{i \in \mathcal{N}} u^k_i(t) + \sum_{i \in \mathcal{N}, T_i}\, ^{T_i}p^k_i(t) = 1.
\label{eq:construp}
\end{align}
}
\vspace{-2mm}

Finally, the following constraint between $u^k_i$, $v^k_{ij}$, and $w^k_{ij}$ ensures that a vehicle only performs one task at a time

\vspace{-2mm}
{\small
\begin{align}
\sum_{i \in \mathcal{N}} \Big(u^k_i(t+1) + \sum_{j \in \mathcal{N}} v^k_{ij}(t) + \sum_{j \in \mathcal{N}} w^k_{ij}(t)\Big) \leq 1,
\label{eq:constr1}
\end{align}
}
where the sum is zero when vehicle $k$ is traveling (i.e., $\sum_{i \in \mathcal{N}, T_i \neq 0}\;\allowbreak ^{T_i}p^k_i(t) = 1$). 

The variables $d_{ij}(t)$, $^{T_i}p^k_i(t)$, and $u^k_i(t)$ make up the state of the system, as they completely define the customer demand and the state of all vehicles. Let $x(t)$ be the state vector, that is, the column vector created by reshaping and concatenating $d_{ij}(t)$, $^{T_i}p^k_i(t)$, and $u^k_i(t)$. We define the set of feasible states by $\mathcal{X}$ where

\vspace{-2mm}
{\small\begin{align}
\mathcal{X} := \left\{x = [d_{ij}\; ^{T_i}p^k_i \; u^k_i]^\intercal \Bigg | 
\begin{array}{l}
d_{ij} \in (\mathbb{N} \cup \{0\})^{N^2}, \;  d_{ii} = 0 \\
^{T_i}p^k_i \in \{0,1\}^\mathcal{D},\;  ^{T_i}p^k_i \text{ satisfies } \eqref{eq:pconstraint} \\
u^k_i \in \{0,1\}^{N |\mathcal{V}|}, \; u^k_i \text{ satisfies } \eqref{eq:construp}
\end{array}
 \right \}
\label{eq:xset}
\end{align}}
\vspace{-3mm}

The variables $v^k_{ij}(t)$ and $w^k_{ij}(t)$ make up the control input of the system. Let $u(t)$ be the control vector, that is, the column vector created by concatenating $v^k_{ij}(t)$ and $w^k_{ij}(t)$. Any feasible control $v^k_{ij}$ which sends vehicles to transport customers cannot transport more customers than there are waiting, thus
\begin{align}
\sum_{k \in \mathcal{V}} v_{ij}^k(t) \leq d_{ij}(t) + c_{ij}(t).
\label{eq:constr4}
\end{align}
%\zsmargin{}{Nitpick: Since, through (6) and (8) the feasible set is a function of the state, it might be helpful to write $\mathcal{U}(x(t),t)$}
We collect our system constraints to form the set of feasible controls, $\mathcal{U}(t)$, where
\vspace{-3.5mm}

{\small\begin{align}
\mathcal{U}(t) := \left\{u = [v^k_{ij}\; w^k_{ij}]^\intercal \Bigg | 
\begin{array}{l}
v^k_{ij} \in \{0,1\}^{|\mathcal{V}|N^2}, v^k_{ii} = 0 \\
w^k_{ij} \in \{0,1\}^{|\mathcal{V}|N^2}, w^k_{ii} = 0 \\
u \text{ satisfies \eqref{eq:constr1} and \eqref{eq:constr4}}
\end{array}
 \right \} .
\label{eq:uset}
\end{align}}
Note that since \eqref{eq:constr1} and \eqref{eq:constr4} are time dependent, $\mathcal{U}(t)$ is time dependent. With this formulation, we have modeled an AMoD system (without battery charging or other operational constraints) as a linear system in the form of
\begin{equation}
\vspace{1mm}
x(t+1) = A x(t) + B u(t) + c(t),
\label{eq:linear}
\end{equation}
where $x(t) \in \mathcal{X}$, $u(t) \in \mathcal{U}(t)$, $c(t) = [c_{ij}(t)\;\; \mathbf{0}\;\; \mathbf{0}]^\intercal$, and $A$ and $B$ are the coefficient matrices associated with \eqref{eq:dij}, \eqref{eq:tpki}, and \eqref{eq:uit}. The vector $c(t)$ represents new customers that arrive every time step and constitutes  an  exogenous disturbance for the system.
%Before it was: ``matrix $C$, representing new customers that arrive every time step, is modeled as a disturbance input to the system."

A few comments are in order. First, one may wonder why information about whether a vehicle is rebalancing or ferrying a passenger is not encoded in the state vector. This is because we have assumed that as soon as a customer boards a vehicle, he/she has been serviced and the vehicle is identical to one that is rebalancing (traveling without a customer). The only thing that matters is the time at which the vehicle arrives at its destination. %Second, we have not enforced any steady-state distribution of vehicles at the stations as in \cite{MP-SLS-EF-DR:12,RZ-MP:15a}. In this current formulation, rebalancing is done only as needed. \frmargin{In Section \ref{sec:extensions} we discuss why it may be useful to have knowledge or an estimate of such distribution.}{TODO:  change reference to point at implementation section}
Second, we have assumed that each station has sufficient parking space for as many vehicles as needed. This may be indeed true if the stations are geographical regions and vehicles are loitering within the region while waiting for customers. However, limited parking spaces are a real concern especially if parking spaces serve as charging stations for electric vehicles. In Section \ref{sec:extensions} we discuss how this AMoD framework can easily be extended to include limited parking capacity. In the next section we outline additional considerations associated with electric vehicles, in particular charging constraints (additional operational constraints are discussed in Section \ref{sec:extensions}).

\subsection{Charging constraints}

For AMoD systems using electric vehicles, range is a major concern. To take into account the limited range of each vehicle, we define as an additional state variable the state of charge of each vehicle, $q^k(t) \in [0,1]$. A value $q^k(t) = 1$ means that the batteries are fully charged while $q^k(t) = 0$ means that the batteries are depleted. Vehicles' batteries discharge while driving, and can be charged at the stations while waiting for customers. The capacity of the batteries is limited, so once the batteries are full (i.e., $q^k(t) = 1$), charging stops. Each vehicle's charge evolves according to

\vspace{-2mm}
{\small
\begin{align}
q^k(t+1) = &\min\{q^k(t) + \alpha_c \sum_{i \in \mathcal{N}} u^k_i(t+1), 1\} \notag \\
&-\alpha_d \sum_{i \in \mathcal{N}, T_i}\ ^{T_i}p^k_i(t+1),
\label{eq:qk}
\end{align}
}
%\begin{align}
%q^k(t) = q^k(t-1) + \alpha_c \sum_{i \in \mathcal{N}} u^k_i(t-1) - \alpha_d b^k(t-1),
%q^k(t) = q^k(t-1) + \alpha_c \sum_{i \in \mathcal{N}} u^k_i(t-1) - \alpha_d \sum_{i, T_i}\,^{T_i}p^k_i(t-1),
%\end{align}
where $\alpha_c>0$ is the rate of charge at a charging station and $\alpha_d>0$ is the rate of discharge while driving.

The charge of the vehicle restricts its range, and in some scenarios, a vehicle may have to wait at its charging station to charge rather than to transport a waiting customer. The charging constraints ensure that each vehicle has enough charge to complete its trip:
\begin{align}
q^k(t) &\geq v^k_{ij}(t)\, \alpha_d\, t_{ij},\label{eq:ccv}\\
q^k(t) &\geq w^k_{ij}(t)\, \alpha_d\, t_{ij}. \label{eq:ccw}
\end{align}
Constraint \eqref{eq:ccv} ensures enough charge for a customer trip and \eqref{eq:ccw} ensures enough charge for a rebalancing trip. %The notation used in this paper is summarized in Table \ref{tab:symbols}.
% [ZS] Dang, this table was super useful, too bad there's not enough room
%\begin{table}
%\centering
%\caption{Table of definitions}
%\begin{tabular}{| c | l |}
%\hline
% & Definitions \\ \hline
%$N$, $|\mathcal{N}|$ & number of stations \\ \hline
%$m$, $|\mathcal{V}|$ & number of vehicles \\ \hline
%$d_{ij}(t)$ & number of customers waiting at station $i$, going to station $j$ at time $t$ \\ \hline
%$c_{ij}(t)$ & number of customer that arrival at station $i$, going to station $j$, at time $t$ \\ \hline
%$t_{ij}$ & travel time from station $i$ to station $j$ \\ \hline
%$^{T_i}p^k_i(t)$ & vehicle $k$ is $T_i$ time steps away from station $i$ at time $t$ \\ \hline
%$u^k_i(t)$ & vehicle $k$ waits at station $i$ from time $t-1$ to $t$ \\ \hline
%$v^k_{ij}(t)$ & vehicle $k$ services a customer going from station $i$ to $j$ at time $t$ \\ \hline
%$w^k_{ij}(t)$ & vehicle $k$ rebalances (drives itself) from station $i$ to $j$ at time $t$ \\ \hline
%$q^k(t)$ & the state of charge of vehicle $k$ at time $t$ \\ \hline
%$\alpha_c$ & battery charge rate \\ \hline
%$\alpha_d$ & battery discharge rate \\ \hline
%\end{tabular}
%\label{tab:symbols}
%\end{table}

\subsection{Objectives}
The primary objective is to service all of the waiting customers as quickly as possible. A secondary goal is to ensure that rebalancing is done in an efficient manner and that vehicles do not rebalance when not necessary to avoid adding congestion on the road. Hence, for each time step $t$ we have the cost functions
\vspace{0.8mm}
\begin{equation}
J_x(x(t)) = \sum_{i,j \in \mathcal{N}} d_{ij}(t), \; \text{primary objective},
\label{eq:objx}
\end{equation} \vspace{-1mm}
\begin{equation}
J_u(u(t)) = \sum_{k \in \mathcal{V}} \sum_{i,j \in \mathcal{N}} t_{ij} w^k_{ij}(t), \; \text{secondary objective}.
\label{eq:obju}
\end{equation}

%In order to ensure good closed-loop performance, it is desirable to align the distribution of vehicles at the end of the optimization horizon with customer demand at that time. 
%At the same time, promoting a ``balanced'' vehicle distribution decreases the time horizon needed for MPC to achieve stability, and hence reduces computational effort. To specify a vehicle distribution, we keep track of the number of vehicles at or enroute to station $i$ (this is used in \cite{MP-SLS-EF-DR:12,RZ-MP:15a}) as $r^{\text{own}}_i(t)$, where
%\[
%r^{\text{own}}_i(t) =  \sum_{k \in \mathcal{V}} \left[ u^k_i(t) + \sum_{T_i}\  ^{T_i}p_i^k(t) + \sum_{j\in\mathcal{N}} \left( v^k_{ji}(t) + w^k_{ji}(t) \right)\right]
%\]
%The following cost function penalizes deviations from a prescribed final vehicle distribution $\hat r^f$
%\begin{equation}
%J_f(x(t_{\text{hor}})) = \sum_{i\in\mathcal N} \left| r^{\text{own}}_i(t_\text{hor}) - \hat r^f _i \right|
%\end{equation}
%While the cost function $J_f(x(t_{\text{hor}}))$ is not linear, it is trivial to convert it to a linear function with $n$ auxiliary variables.

When charging constraints are considered, we may include the vehicles' final state of charge as an objective to maximize. This allows us to trade off short-term quality-of-service and long-term battery capacity. To this end, we define the cost function
\begin{equation}
J_c(x(t_{\text{hor}})) =  \sum_{k \in \mathcal{V}} q^k(t_\text{hor}).
\end{equation}
%Denote this quantity by $r^{\text{own}}_i(t)$, which can be propagated as a state variable as follows
%\vspace{-2mm}
%{\small\begin{equation}
%r^{\text{own}}_i(t+1) = r^{\text{own}}_i(t) + \sum_{k \in \mathcal{V}} \sum_{j \in \mathcal{N}} (v^k_{ji}(t) + w^k_{ji}(t))
%- \sum_{k \in \mathcal{V}} \sum_{j \in \mathcal{N}} (v^k_{ij}(t) + w^k_{ij}(t)). \notag
%\vspace{-3mm}
%\end{equation}}

\section{Problem Formulation} \label{sec:problem}
The objective of this paper is to design model predictive control algorithms that are 1) provably stable in the sense of Lyapunov
%}{Say that this is the first thing we typically study in MPC  - Marco. Stress this is a MINIMAL requirement, but still quite important. Stress this is required to keep waiting times well regulated.}
 and 2) robust against the exogenous disturbance $c(t)$ (customer arrivals). 
 We can now rigorously formulate the first problem, namely the AMoD regulation problem:  
\begin{quote}
{\normalsize{\bf AMoD Regulation Problem (ARP)}: Assume $c_{ij}(t) = 0$ for all time $t > 0$. For each time $t$, select feasible control inputs $u(t) \in \mathcal{U}(t)$ such that as $t \rightarrow \infty$, $J_x(x(t)) \rightarrow 0$. 
}
\end{quote}
In the definition of the ARP it is assumed that the exogenous disturbance $c(t)$ is identically equal to zero for $t>0$, in other words no new customers arrive after time zero. Hence, the ARP captures the \emph{minimal} requirement that an initial set of customers is eventually transported to the respective destinations, under the assumption of zero future customer arrivals---hence the name ``regulation problem."

%The study of the ARP is of interest because the requirement that $J_x(x(t)) \rightarrow 0$ guarantees that the passengers' wait times will remain well-regulated.
While the ARP can also be solved by straightforward algorithms such as nearest neighbor dispatch (presented in Section \ref{sec:results}), it is nevertheless critical to show that our MPC algorithm does not only offer good real-world performance but also guarantees analytical stability.
The second objective, robustness against the exogenous disturbance $c(t)$, is analyzed in simulation in Section \ref{sec:results}.

%\frmargin{The objective of this paper is to design model predictive control algorithms that are provably stable in the sense of Lyapunov (i.e., solve the ARP problem) and}{} are \frmargin{robust against}{Move some form of this sentence to just before problem definition} the exogenous disturbance $c(t)$ (shown in simulation). 
%Note that by the discrete nature of $J_x(x(t))$ ($J_x(x(t)) \in \mathbb{N} \cup \{0\}$), \mpmargin{it will converge to zero in finite time}{Confusing/unclear, is it important?}. Note also 
In the ARP, we note that $J_x(x(t)) \rightarrow 0$ will cause $J_u(u(t)) \rightarrow 0$ which implies that $u(t) \rightarrow \mathbf{0}$ by the definitions of $J_x(x(t))$ and $J_u(u(t))$. 

Before presenting our MPC algorithms, we show some important structural properties of our problem setup. We first show that given $x(t) \in \mathcal{X}$, and $u(t) \in \mathcal{U}(t)$, the state of the undisturbed system at the next time step, $x(t+1) = Ax(t) + Bu(t)$, automatically satisfies \eqref{eq:pconstraint} and \eqref{eq:construp}, and thus $x(t+1) \in \mathcal{X}$. This property ensures the \emph{persistent feasibility} of the MPC algorithms presented in Section \ref{sec:mpc}.
\ifarxiv
The proof of the following proposition can be found in the appendix.
\else
The proof of the following proposition can be found in \cite{RZ-FR-MP:16EV}.
\fi
%\mpmargin{}{Well, the theorem is proven for the case $t=0$, where $\mathcal U = \mathcal U(0)$, how about the future steps?}
\begin{proposition}[Feasible Sets] \label{prop:feasible}
Given $x \in \mathcal{X}$, $u \in \mathcal{U}(t)$, and $x^+$ given by $x^+ = Ax + Bu$, then $x^+ \in \mathcal{X}$. 
\end{proposition}

With this result we can be sure that the reachable space of the system is feasible. 

\section{Model Predictive Control of AMoD} \label{sec:mpc}
In this section we present two MPC algorithms to optimize vehicle scheduling and routing in an AMoD system. Specifically, the first MPC algorithm addresses the case without charging constraints (Section \ref{subsec:no_charge}), while the second MPC algorithm allows the inclusion of charging constraints (Section \ref{subsec:yes_charge}). We prove that both algorithms solve the ARP (our technical approach is to prove asymptotic stability in the sense of Lyapunov). Note that in general asymptotic Lyapunov stability is a stronger result than simply proving $J_x(x(t))\to 0$.
However, due to the boundedness of the number of customers ($d_{ij}(t)$), asymptotic Lyapunov stability coincides with solving the ARP in this case. We remark that proving asymptotic stability in the sense of Lyapunov does not only guarantee that the number of passengers will decrease to zero (hence, wait times will not grow unbounded) but also implies that, if initial conditions are small (i.e. few passengers are requesting service), wait times will also be small. We numerically characterize the performance of the MPC algorithms in Section \ref{sec:results} (in particular, we study their ability to deal with a continuous stream of arriving customers).

\subsection{MPC without charging constraints}\label{subsec:no_charge}

In this section we present the MPC algorithm for solving the AMoD regulation problem without charging constraints. In an MPC algorithm, an optimization problem is solved at each time instant giving a sequence of control actions up to a time horizon $t_\text{hor}$. The first step of the control sequence is implemented and the system is re-optimized at the next time instant. Let $u(t+k)_{|t}$ be the control action at time $t+k$, solved at time $t$, where $k \in \{0, t_{\text{hor}}-1\}$. 

%\mpmargin{}{Add a sentence about the MPC idea, i.e., its receding horizon implementation. Strictly speaking you currently have only an optimization problem Also, in the introduction, it may be a good idea to highlight the practical advantages of an MPC algorithm, e.g., its built-in closed-loop behavior}. 
%\mpmargin{}{Clarify the notation $u(t)_{|t}$}
\begin{alg}[MPC without charging constraints]\label{alg:mpc1}
Given \allowbreak$x(t) \in~\mathcal{X}$, at each time instant $t \in \mathbb{N}$ the controls $u(t)_{|t}, u(t+1)_{|t},\ldots, u(t+t_{\text{hor}}-1)_{|t}$ are obtained by solving the optimization problem

\vspace{-2mm}
{\small \begin{align}
\underset{u(t),\ldots,u(t+t_{\text{hor}}-1)}{\text{minimize}}  & \sum_{\tau = t}^{t+t_{\text{hor}}-1} 
J_x(x(\tau+1)) + \rho_1 J_u(u(\tau)) \notag \\
% + \rho_f J_f(x(t_\text{hor})) \notag \\
 \quad \text{subject to} \quad\;\; & x(\tau+1) = Ax(\tau) + Bu(\tau) \notag \\
 & x(\tau+1) \in \mathcal{X} \notag \\
 & u(\tau) \in \mathcal{U}(\tau) \notag \\
 & \tau = t,\ldots,t+t_{\text{hor}}-1. \notag
\end{align}
}
where $\rho_1 > 0$ and $J_x(x(\tau))$ and $J_u(u(\tau))$ are given by \eqref{eq:objx} and \eqref{eq:obju}, respectively. Implement $u(t)_{|t}$ and repeat the optimization at the next time instant.
\end{alg}
\begin{remark}
The purpose of $\rho_1 J_u(u(\tau))$ in the objective is to avoid unnecessary vehicle rebalancing. However, since not enough rebalancing may result in customers not receiving service, this term is secondary to the primary objective of servicing customers so $\rho_1$ should be set to a small value.
\end{remark}

We are now ready to present the main result of this section, which shows that Algorithm \ref{alg:mpc1} solves the ARP problem.
%In the next two theorems, we show that if the time horizon $t_{\text{hor}}$ is one, the MPC algorithm \ref{alg:mpc1} produces a solution that is stable in the sense of Lyapunov. If $t_{\text{hor}} > \max_{i,j} t_{ij}$, then algorithm \ref{alg:mpc1} produces an asymptotically stable solution which solves Problem \ref{problem1}. 
\begin{theorem}[Asymptotic stability of Algorithm 1]\label{thm:mpc1}
Suppose $t_{\text{hor}} \geq 2\max_{i,j \in \mathcal{N}} t_{ij}$. Then Algorithm \ref{alg:mpc1} solves the AMoD regulation problem.
\end{theorem}
The proof of Theorem \ref{thm:mpc1}  
\ifarxiv
is reported in
\ifimages 
\else
the 
\fi
Appendix \ref{apx:proofs}.
\else
can be found in the Appendix of \cite{RZ-FR-MP:16EV}.
\fi
The key idea of the proof is to show that at least one customer is serviced every $t_\text{hor}$ time steps. We can do this by defining a new linear system equivalent to \eqref{eq:linear} where $t_\text{hor}$ time steps in \eqref{eq:linear} correspond to one time step in the new system. We can then use an extension of Lyapunov stability for set-valued functions (\cite[Theorem 4]{LM:05}) to prove the asymptotic stability of the system. %We can then verify that the system meets the conditions for Once certain technical conditions are verified, asymptotic stability of the system in the sense of Lyapunov follows from \cite[Theorem 4]{LM:05}.

\subsection{MPC with charging constraints}\label{subsec:yes_charge}
In this section we extend the results in the previous section to account for range limitations and charging constraints associated with electric vehicles. To do this, we first augment the state vector $x(t)$ with the charge of each vehicle, $q^k(t)$. The new state vector becomes $x' = [d_{ij}\; ^{T_i}p^k_i\; u^k_i\; q^k]^{\intercal}$. The state $q^k(t)$ is propagated at each time step according to \eqref{eq:qk}. However, \eqref{eq:qk} is piecewise linear (due to the $\min$ operator) so the extended system cannot be written in the form of \eqref{eq:linear}. As we will see, this will not be an issue for the MPC algorithm. Finally, we add the range constraints \eqref{eq:ccv} and \eqref{eq:ccw} to the definition of $\mathcal{U}(t)$. To summarize, the augmented feasible states and controls are
\vspace{-1mm}

{\small\begin{align}
\mathcal{X}' = \left\{x' = [x\; q^k]^\intercal \bigg | 
\begin{array}{l}
x \in \mathcal{X} \\
%d_{ij} \in (\mathbb{N} \cup \{0\})^{N^2}, d_{ii} = 0 \\
%^{T_i}p^k_i \in \{0,1\}^\mathcal{D},\ ^{T_i}p^k_i \text{ satisfies } \eqref{eq:pconstraint} \\
q^k \in \mathbb{R}^\mathcal{V}, \, \, 0 \leq q^k \leq 1
\end{array}
 \right \},
\label{eq:xset2}
\vspace{-1mm}
\end{align}}
{\small\begin{align}
\mathcal{U}'(t) = \left\{u' = [v^k_{ij}\; w^k_{ij}]^\intercal \Bigg | 
\begin{array}{l}
v^k_{ij} \in \{0,1\}^{|\mathcal{V}|N^2}, \quad v^k_{ii} = 0 \\
w^k_{ij} \in \{0,1\}^{|\mathcal{V}|N^2}, \quad w^k_{ii} = 0 \\
u' \text{ satisfies \ref{eq:constr1}, \ref{eq:constr4}, \ref{eq:ccv}, and \ref{eq:ccw}}
\end{array}
 \right \}.
\label{eq:uset2}
\end{align}}
\vspace{-2mm}

We turn our attention back to $q^k(t)$ and notice that from \eqref{eq:qk}, $q^k(t+1)$ satisfies the following two linear inequalities

{\small \begin{align}
q^k(t+1) &\leq q^k(t) + \alpha_c \sum_{i \in \mathcal{N}} u^k_i(t+1) - \alpha_d \sum_{i \in \mathcal{N}, T_i}\ ^{T_i}p^k_i(t+1) \label{eq:qineq1} \\
q^k(t+1) &\leq 1 - \alpha_d \sum_{i \in \mathcal{N}, T_i}\ ^{T_i}p^k_i(t+1). \label{eq:qineq2}
\end{align}
}
Equation \eqref{eq:qk} can be satisfied in our MPC algorithm by satisfying \eqref{eq:qineq1} and \eqref{eq:qineq2} and maximizing $q^k$. 
\begin{alg}[MPC with charging constraints]\label{alg:mpc2}
At each time instant $t \in \mathbb{N}$ the controls $u'(t)_{|t}, u'(t+1)_{|t},... u'(t+t_{\text{hor}}-1)_{|t}$ are obtained by solving the optimization problem
\vspace{-2mm}

{\small\begin{align}
\underset{u'(t),...,u'(t+t_{\text{hor}}-1)}{\text{minimize}}  & \sum_{\tau = t}^{t+t_{\text{hor}}-1} \bigg (
J_x(x(\tau+1)) + \rho_1 J_u(u'(\tau))  \notag \\
& - \rho_2 \sum_{k \in \mathcal{V}} q^k(\tau+1) \bigg )
% + \rho_f J_f(x(t_\text{hor}))
  - \rho_c J_c(x(t_\text{hor}))   \notag \\
  \text{subject to} \;\;\quad & x(\tau+1) = Ax(\tau) + Bu(\tau) \notag \\
 & q^k(\tau+1) \leq q^k(\tau) + \alpha_c \sum_{i \in \mathcal{N}} u^k_i(\tau+1) - \notag \\ & \alpha_d \sum_{i \in \mathcal{N}, T_i}\ ^{T_i}p^k_i(\tau+1) \notag \\
 & q^k(\tau+1) \leq 1 - \alpha_d \sum_{i \in \mathcal{N}, T_i}\ ^{T_i}p^k_i(\tau+1) \notag \\
 & x'(\tau+1) \in \mathcal{X}' \notag \\
 & u'(\tau) \in \mathcal{U}'(\tau) \notag \\
 & \tau = t,\ldots,t+t_{\text{hor}}-1 \notag
\end{align}}
where $\rho_1 > 0$, $\rho_2 > 0$, and $\rho_c>0$. Implement $u'(t)_{|t}$ and repeat the optimization at the next time instant.
\end{alg}

The next theorem shows that Algorithm 2 solves the ARP with charging constraints.  

\begin{theorem}[Asymptotic stability of Algorithm 2]\label{thm:mpc2} 
Suppose $t_{\text{hor}} \geq 2(1+\frac{\alpha_d}{\alpha_c})\max_{i,j \in \mathcal{N}} t_{ij} $. Then Algorithm \ref{alg:mpc2} solves the AMoD regulation problem with charging constraints.
\end{theorem}
The proof for this theorem follows the same procedure as the proof of Theorem \ref{thm:mpc1} and
\ifarxiv
is also reported in Appendix \ref{apx:proofs}.
\else
can also be found in \cite{RZ-FR-MP:16EV}.
\fi
As in Theorem \ref{thm:mpc1}, we define a time-scaled version of \eqref{eq:linear}. The key difference is to note that a vehicle may be completely depleted of charge after a trip and requires $({\alpha_d}/{\alpha_c} )t_{ij}$ time steps to charge before departing on its next trip. We can further observe that fast charging reduces the time horizon needed to maintain stability.
%\begin{proof}
%The proof for this theorem follows the same procedure as the proof of Theorem \ref{thm:mpc1}. The main difference is the choice of the time horizon. Consider again the case where vehicle $k$ is enroute from station $j$ to station $i$. Suppose there are no customers at station $i$, so once the vehicle arrives, it must travel back to station $j$ for the next pickup. However, once the vehicle arrives at station $i$, it is completely depleted of charge and needs to charge for $\frac{\alpha_d}{\alpha_c} t_{ij}$ amount of time before departing. Once the vehicle arrives back at station $j$, it is again depleted of charge and must charge again before servicing the customer. The total time this takes consists of two traveling periods and two charging periods, with the travel time upper bounded by $\max_{i,j \in \mathcal{N}} t_{ij}$ and charging time upper bounded by $\frac{\alpha_d}{\alpha_c} \max_{i,j \in \mathcal{N}} t_{ij}$. This completes the proof.
%\end{proof}
\begin{remark}
The time horizon bounds given in Theorem \ref{thm:mpc1} and \ref{thm:mpc2} assume the worst case scenario, which would very rarely occur in practice. Thus, in most practical cases, a shorter time horizon (which reduces computational cost) should also reduce $d_{ij}$ to zero. In Section \ref{sec:regulation} we will show this is indeed the case.
\end{remark}

In the next section we show through simulation that the MPC algorithms solve the AMoD regulation problem and we benchmark their performance against other algorithms in the literature.
% \frmargin{In Section \ref{sec:results} we show through simulation that Algorithms \ref{alg:mpc1} and \ref{alg:mpc2} are able to solve the AMoD regulation problem by servicing all outstanding customers. Furthermore, we characterize their disturbance rejection properties when servicing a continuous stream of customers and benchmark their performance against other algorithms in the literature.}{A bit inconsequential?}

\section{Simulation Results} \label{sec:results}
In this section, we present three sets of simulation results that demonstrate the correctness and performance of our MPC approach. First, we show through simulation that the AMoD regulation problem can indeed be solved using Algorithm \ref{alg:mpc1} and Algorithm \ref{alg:mpc2}. Second, we show using real taxi data that Algorithm \ref{alg:mpc1} yields real-time performance on small to medium-sized systems and outperforms several state-of-the-art algorithms from the literature in terms of customer wait times. Finally, we study how the charge/discharge rate of batteries affect the performance of an AMoD system with electric vehicles. For all simulations, Algorithm \ref{alg:mpc1} and \ref{alg:mpc2} were implemented using the IBM CPLEX solver for mixed-integer linear programs (MILP) \cite{cplex}. 

%\vspace{-0.2cm}
\subsection{AMoD regulation} \label{sec:regulation}
To validate Algorithms \ref{alg:mpc1} and \ref{alg:mpc2}, an initial $d_{ij}$ was randomly generated with up to 30 customers at each station while $c_{ij}(t)$ was set to zero for all $t$. The simulation was performed with 30 vehicles and 10 stations, with 3 vehicles at each station to begin with. The maximum travel time between two stations was 7 time steps. For the system without charging constraints, $t_{\text{hor}}$ was set to 10 steps while for the system with charging constraints, $t_{\text{hor}}$ was set to 20. 
The weight of the secondary objectives were set to $\rho_1=0.01$ and, for the system with charging constraints, $\rho_2=0.001$, $\rho_c=0$. 
Figure \ref{fig:amodreg1} shows the number of customers waiting at each of the 10 stations as a function of time. Figure \ref{fig:amodregcharge} shows the number of customers vs. time for a system with charging constraints. In this case, the initial charge of all vehicles was set to 0.8 and vehicles could charge twice as fast as they could discharge ($\alpha_c = 0.2$ and $\alpha_d = 0.1$). Figure \ref{fig:amodcharge2v} shows the charge levels of two of the vehicles and illustrates that when future customer demand is not taken into account, a general strategy for each vehicle is to service customers until its batteries are almost depleted, then charge just enough to service the next customer. The need to recharge after trips results in longer wait times and in this case, a longer total time to service all the customers (50 minutes for the case with charging constraints and 30 minutes without).
\vspace{-2mm}
\begin{figure}[h]
\centering
\hspace{-1em}
		\subfigure[]{\label{fig:amodreg1}    \includegraphics[width=.16\textwidth]{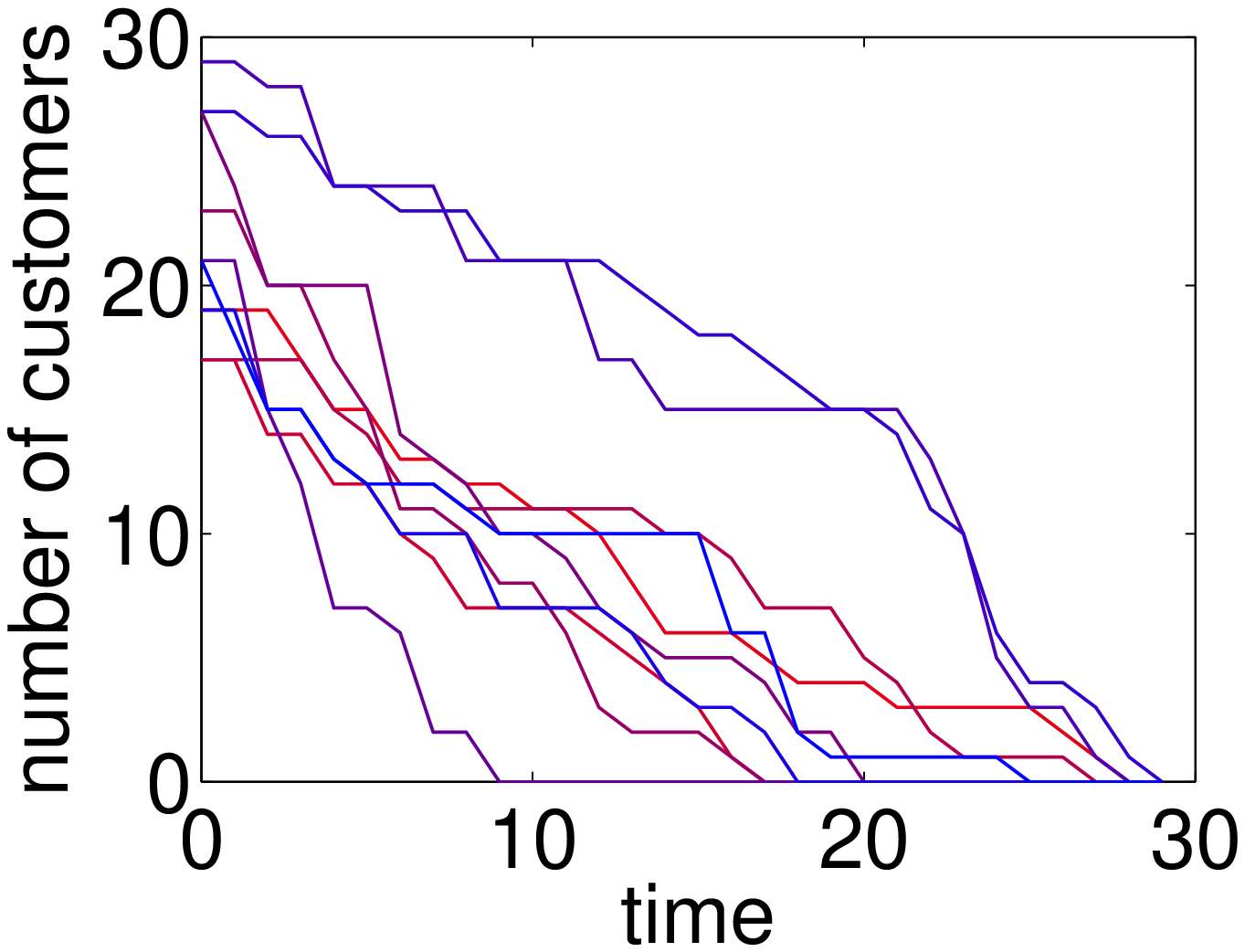}}  \hspace{-1em}
		\subfigure[]{\label{fig:amodregcharge}	\includegraphics[width=0.16\textwidth]{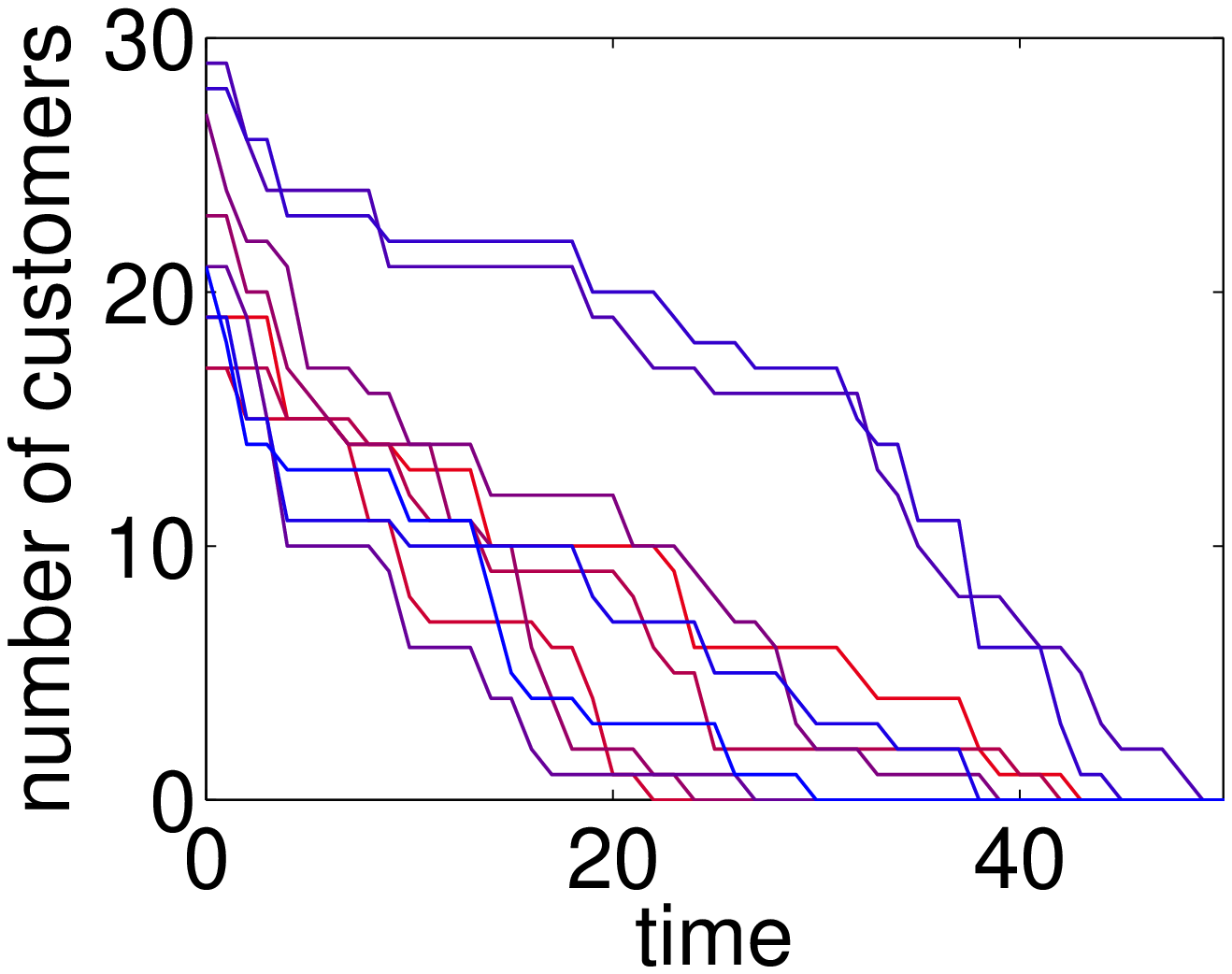}} \hspace{-1em}
		\subfigure[]{\label{fig:amodcharge2v}	\includegraphics[width=0.16\textwidth]{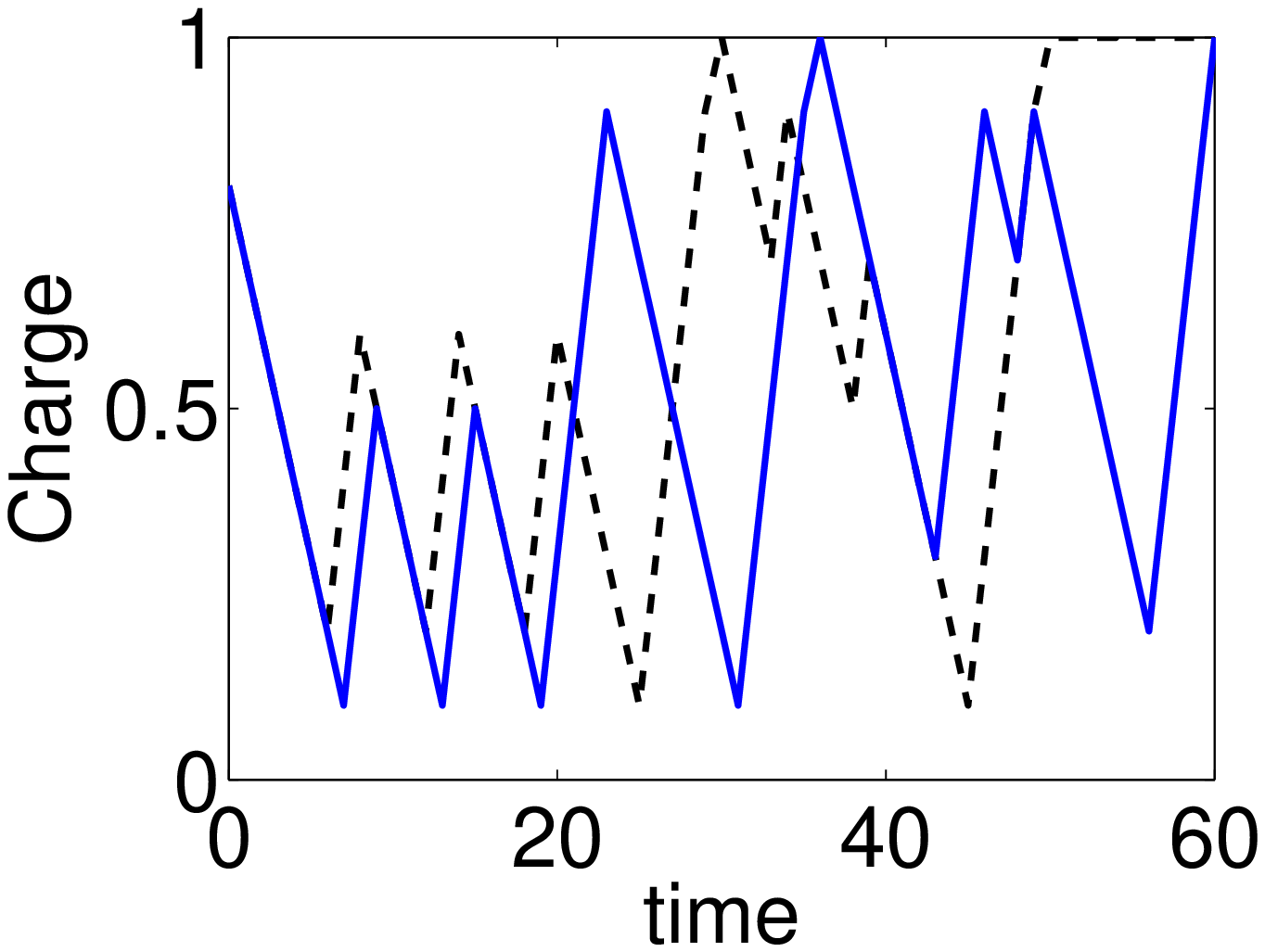}}
\vspace{-0.5em}
\caption{\ref{fig:amodreg1}: Number of customers waiting at each station as a function of time for 10 stations, 30 vehicles without charging constraints. \ref{fig:amodregcharge}: Number of customers waiting at each station as a function of time with charging constraints. \ref{fig:amodcharge2v}: State of charge for two vehicles as a function of time.}
\label{fig:amodreg}
\vspace{-8mm}
\end{figure}
%\vspace{-2mm}

%The simulations are performed for a small-scale AMoD system akin to pilot projects currently under development \cite{Google:14}. The system parameters used in the simulations are based on real taxi data from New York City\footnote{Courtesy of the New York City Taxi \& Limousine Commission}. Taxi trips within New York City's financial district (in Lower Manhattan) are aggregated into 10 stations using \emph{k}-means clustering. The customer arrival rates, routing probabilities, and travel times are estimated based on the trip data. The maximum travel time between any two stations was determined to be 7 minutes. All simulations use a time step of one minute. 
%\vspace{-0.2em}
\subsection{Performance of MPC} \label{sec:sim_performance}
To evaluate the performance of our MPC algorithm, we conducted an extensive simulation study comparing Algorithm \ref{alg:mpc1} to several other AMoD and taxi dispatch algorithms found in the literature using real New York taxi data\footnote{Courtesy of the New York City Taxi \& Limousine Commission}. We show that Algorithm \ref{alg:mpc1} not only outperforms other algorithms in terms of customer wait times, but can be used as an ``optimal'' baseline to quantitatively evaluate the performance of other taxi dispatch or AMoD algorithms. The simulations are performed with 40 vehicles for 24 hours with a time step of 6 seconds. Taxi trips within New York City's Financial District area (Lower Manhattan, south of Canal St.) are extracted for nine Mondays in March and April of 2012, resulting in 2300-3500 trips per day. The simulated vehicles move along the Manhattan distance between pickup and drop-off locations, with speeds estimated from the data to account for congestion. 
The Financial District is divided into 15 regions: the center of each region (a ``station'') is computed using \emph{k}-means clustering on historical customer origin/destination data. In the simulation, customers are \emph{not} required to go to a station to receive service: once a vehicle is assigned to a customer, it drives to the customer's location and drops her/him off at the requested destination, then drives to the nearest station. The role of stations is to 1) model the availability of parking and charging facilities and 2) provide a discretized model for the vehicle rebalancing problem.

For this simulation study, we implemented six dispatch algorithms, including two versions of Algorithm \ref{alg:mpc1}:
\begin{enumerate}
\item \emph{Nearest-neighbor dispatch (NN):} Each customer is assigned the nearest free vehicle. If no vehicles are available, the customer request is added to a first-in, first-out queue. Free vehicles move according to a random walk until assigned to a customer. 
%This algorithm was used in \cite{KTS-NHD-DHL:10} as a baseline algorithm...
\item \emph{Collaborative dispatch (CD) \cite{KTS-NHD-DHL:10}:} customer requests are aggregated in a queue (a single queue is used for the entire Financial District) and, when the queue size reaches a threshold $x$, the same number of free vehicles are dispatched to the customers. Vehicles are matched to customers to minimize the total distance they need to drive empty. The algorithm is modified from \cite{KTS-NHD-DHL:10} in two ways: 1) a time-varying queue size $x(t)$ is employed to account for highly time-varying demand, and 2) the optimization is solved in a centralized fashion. 
\item \emph{Markov redistribution (MR) \cite{MV-JA-DR:12}:} Customer demand information is used to rebalance empty vehicles among stations in order to drive the vehicles' distribution towards the distribution of passenger arrivals. The problem is cast as a linear program (LP) and yields a randomized rebalancing strategy for each empty vehicle. The algorithm implemented in our simulation is modified from \cite{MV-JA-DR:12} in two ways: 1) the problem is solved exactly as an LP and 2) in order to accommodate time-varying demand, the algorithm rebalances vehicles based on the sum of (i) estimated future customer demand and (ii) current number of passengers waiting. Since the MR algorithm is randomized, its performance can vary widely between trials: for each day, the results presented are the median of ten executions.

\item \emph{Real-time rebalancing (RR) \cite{MP-SLS-EF-DR:12}:} This algorithm rebalances vehicles based on current waiting customers. At each rebalancing epoch (every 2 minutes) the algorithm computes the number of vehicles at or enroute to each station and solves a linear program to evenly distribute excess free vehicles throughout the system. It uses the same fifteen stations as the MR algorithm. 
\item \emph{Algorithm 1, MPC with sampled customer arrivals (MPCS):} The algorithm uses the same 15 stations employed by MR and solves the problem with a time horizon of 15 minutes. Customer arrival rates, computed using historical data, are sampled as a Poisson process and fed into Algorithm \ref{alg:mpc1} as predicted future arrivals ($c_{ij}$). This sampling is done every 2 minutes to prevent unnecessary rebalancing. The rebalancing weight is set to $\rho_1=0.01$. A small term is added to the cost function to promote a uniform distribution of vehicles among the stations at the end of the optimization horizon.
\item \emph{Algorithm 1, MPC with full arrival information (MPCF):} The actual customer arrivals over the time horizon (15 minutes) are fed into Algorithm \ref{alg:mpc1} as $c_{ij}$. The algorithm optimizes vehicle assignments with perfect knowledge of the next 15 minutes, and can therefore serve as a baseline for optimal performance, as long as wait times are small compared to the optimization horizon. The rebalancing weight is set to $\rho_1=0.01$. %, as in the MPCS algorithm.
\end{enumerate}

Table \ref{tab:peak} shows the peak customer wait times for each algorithm over the nine days that were simulated (every Monday from March 5 to April 30, 2012). The algorithm that yielded the best performance for each day (shortest peak wait time) is shown in bold (MPCF was excluded since it is non-causal). We note that in all the days where the peak wait time for MPCF is less than the optimization horizon of 15 minutes (day 3, 5, 6, and 9), MPCF achieves the best performance and can effectively serve as an optimal baseline. % \frmargin{(when the average wait time exceeds the optimization horizon, MPC algorithms are unable to effectively use information on future arrivals because new passengers are generally not served during the optimization horizon.
%; thus, in this regime, the MPCF algorithm occasionally has worse performance than MPCS and RR)}{Rick, is this accurate?}.
When the customer demand is high and the wait time is greater than the optimization horizon, the performance of MPC algorithms suffers: MPCS and MPCF are unable to effectively use information on future arrivals because new passengers are generally not serviced within the optimization horizon. Thus, in this regime, the MPCF algorithm is not a reliable optimal baseline. Nevertheless, MPCS achieves the best performance in 7 of the 9 days simulated. It's also worth noting that over the 4 days with short wait times, MPCS achieved peak wait times that were on average 34\% shorter than the next best algorithm, RR. Figure \ref{fig:amodWaitTime} shows the simulation results for day 5 (April 2). In addition to achieving a lower peak wait time, the MPC algorithms are also able to recover quickly to service the remaining customers after peak demand. This is illustrated in Table \ref{tab:time50}, which lists the fraction of time spent for each algorithm where the average wait time was at least 50\% of the respective peak. In this respect, MPCS again consistently outperforms the other algorithms.

A few observations about the other algorithms are in order. First, performance of the CD algorithm is generally slightly better than the NN algorithm, and is consistent with the results obtained by the authors in \cite{KTS-NHD-DHL:10}. However, the problem of selecting a threshold queue size (i.e. the algorithm's tuning parameter) to maximize performance remains open. As a matter of fact, in three instances the NN algorithm outperforms CD in our simulations. Second, as neither NN nor CD rebalance empty vehicles so as to anticipate future demand, their performance highlights the critical importance of preemptive routing to achieve good quality-of-service. Third, since the MR algorithm is a randomized algorithm, its performance can be (and occasionally is) very suboptimal and, at times, significantly worse than NN. Median performance (over 10 runs), however, is generally better than both NN and CD. Finally, MR is conceived for steady-state systems. Thus, while the algorithm can be extended to time-varying systems, it is unclear whether the performance presented in \cite{MV-JA-DR:12} can be replicated in scenarios with highly variable customer demand.
\begin{figure}[h]
\vspace{-1em}
\centering
\includegraphics[width=.47\textwidth]{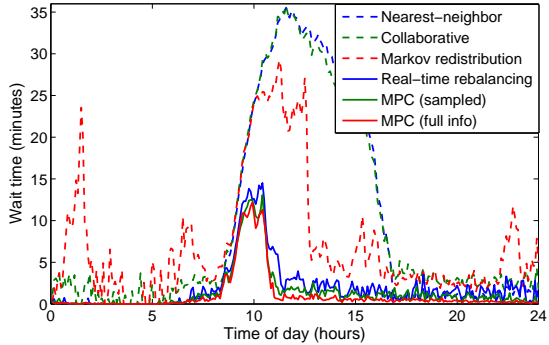} \vspace{-0.5em}
\caption{Average customer wait times throughout the day (April 2, 2012) for all dispatch algorithms.}
\label{fig:amodWaitTime}
\vspace{-1.4em}
\end{figure}
\begin{table}[h]
\centering
\caption{Peak wait time in minutes for each algorithm}
\begin{tabular}{ l c c c c c c c c c }
Day  & 1 & 2 & 3 & 4 & 5 & 6 & 7 & 8 & 9 \\ \hline
NN   & 36 		& 34 		& 27 		& 56 		& 36 		& 9 		& 38 		& 41 		& 24 \\
CD   & 35 		& 34 		& 24 		& 56 		& 36 		& 10 		& 39 		& 41 		& 25 \\
MR   & 37 		& 53 		& 11 		& 31 		& 29 		& 13 		& 32 		& 35 		& 16 \\
RR   &\textbf{16}& 16 		& 7 		& 20 		& 15 		& 7 		& 27 		&\textbf{21}& 10 \\
MPCS & 17 		&\textbf{15}&\textbf{4}	&\textbf{18}&\textbf{13}&\textbf{3}	&\textbf{24}& 24 		&\textbf{7} \\
MPCF* & 18 		& 19 		& 3 		& 19 		& 12 		& 2 		& 29 		& 24 		& 6
\label{tab:peak}
\end{tabular}
\end{table}
\vspace{-2mm}
\setlength{\tabcolsep}{4pt}
\begin{table}[h]
\vspace{-.5em}
\centering
\caption{Fraction of time where the average wait time was at least 50\% of peak}
\begin{tabular}{ l c c c c c c c c c }
Day  & 1 & 2 & 3 & 4 & 5 & 6 & 7 & 8 & 9 \\ \hline
NN   & 0.15 		& 0.16 		& 0.18 		& 0.34 		& 0.27 		& 0.07 		& 0.17 		& 0.19 		& 0.12 \\
CD   & 0.15 		& 0.16 		& 0.16 		& 0.35 		& 0.26 		&\textbf{0.07}& 0.17 		& 0.19 		& 0.11 \\
MR   & 0.09 		& 0.06 		& 0.12 		& 0.09 		& 0.13 		& 0.08 		& 0.10 		& 0.08 		& 0.09 \\
RR   & 0.08			& 0.07 		& 0.11 		& 0.09 		& 0.07 		& 0.07 		&\textbf{0.07}& 0.06		& 0.05 \\
MPCS &\textbf{0.06}	&\textbf{0.05}&\textbf{0.09}&\textbf{0.08}&\textbf{0.06}& 0.17\tablefootnote{Note that the peak wait time for MPCS is only 3 minutes} & 0.08	&\textbf{0.06}&\textbf{0.04} \\
MPCF* & 0.07 		& 0.04 		& 0.06 		& 0.08 		& 0.06 		& 0.08 		& 0.07 		& 0.06 		& 0.03
\label{tab:time50}
\end{tabular}
\vspace{-1em}
\end{table}

The median runtime per iteration of Algorithm \ref{alg:mpc1} on a 2.8 GHz Intel Core i7 PC with 16GB of RAM was 5.5 seconds. The other algorithms analyzed were significantly faster: the NN, CD, MR and RR algorithms had a median runtime per iteration of 2.2, 0.3, 18 and 332 ms respectively. 
Nevertheless, our results show that Algorithm \ref{alg:mpc1} is amenable to a real-time implementation for a moderately-sized system.

\subsection{Effect of charge rate}
In this section we explore the performance limitations of AMoD systems with electric vehicles. Specifically, we would like to answer the question: how does charging rate affect the ability of an AMoD system to service customer demand? To this end, simulations were performed using 40 vehicles with taxi data from 7 am to 3 pm (period of high demand) on April 2, 2012. The simulations were performed for the MPCS algorithm (see Section \ref{sec:sim_performance}) with charging constraints (Algorithm \ref{alg:mpc2}). The discharge rate, $\alpha_d$, was chosen to be $0.0037$, which corresponds to an electric vehicle such as a Nissan Leaf or BMW i3 driving at 20 km/h while using 70\% of its battery capacity (to avoid over-discharging, which could damage the batteries). Three charging rates were used: $\alpha_c = \alpha_d$, $\alpha_c = 2\alpha_d$, and $\alpha_c=4\alpha_d$, which correspond to a charging time of 4 hours, 2 hours, and 1 hour, respectively. The charging times represent realistic current electric vehicle charging capabilities. To avoid prematurely depleting the batteries, a higher final charging cost $\rho_c$ was assigned to the cases with slower charge rates. Figure \ref{fig:chargeWaitTime} shows the customer wait times for the four simulations performed, and \ref{fig:chargeState} shows the average state-of-charge of vehicles over time. 

We first note that as long as there is excess battery capacity in the system, the customer wait times for an electric AMoD system are comparable to the AMoD system without charging constraints (see Figure \ref{fig:amodWaitTime}). In this case, a charge rate of $4\alpha_d$ is stabilizing and is able to support future demand. A charge rate of $2\alpha_d$ is able to service demand without increasing wait time, but the batteries are almost fully depleted by the end of the simulation. A charge rate of $\alpha_d$ is too slow to support customer demand for the entire simulation duration. The large final charging cost $\rho_c = 100$ trades off quality of service with battery capacity, resulting in a slightly higher peak wait time. Even so, batteries become depleted near the end of the simulation period and the algorithm begins to prefer battery charging over servicing customers, resulting in longer wait times. 

\vspace{-1mm}
\begin{figure}[h]
\vspace{-0.9em}
\centering
		\subfigure[]{\label{fig:chargeWaitTime}    \includegraphics[width=.235\textwidth]{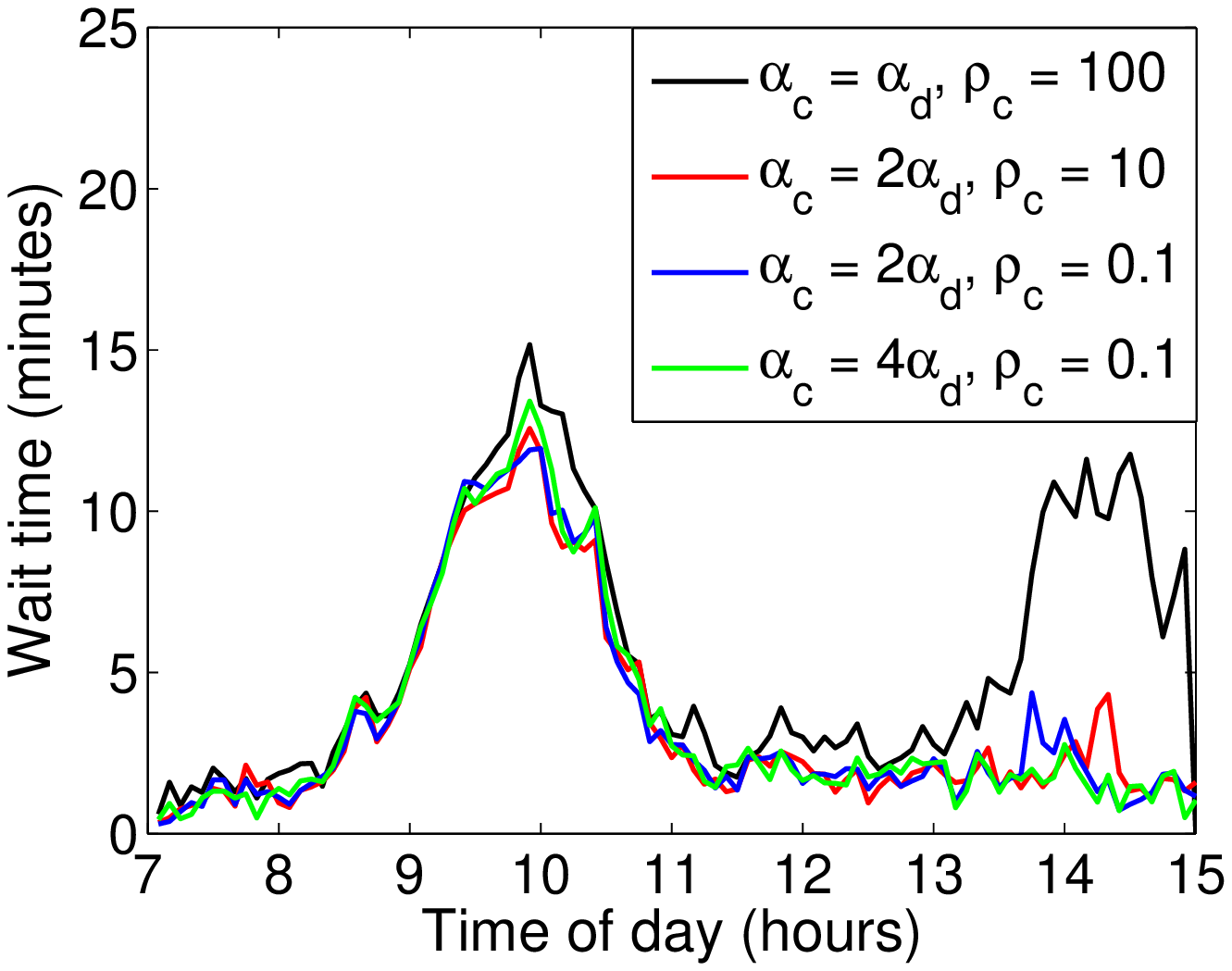}} \hspace{-1em}
		\subfigure[]{\label{fig:chargeState}	\includegraphics[width=0.235\textwidth]{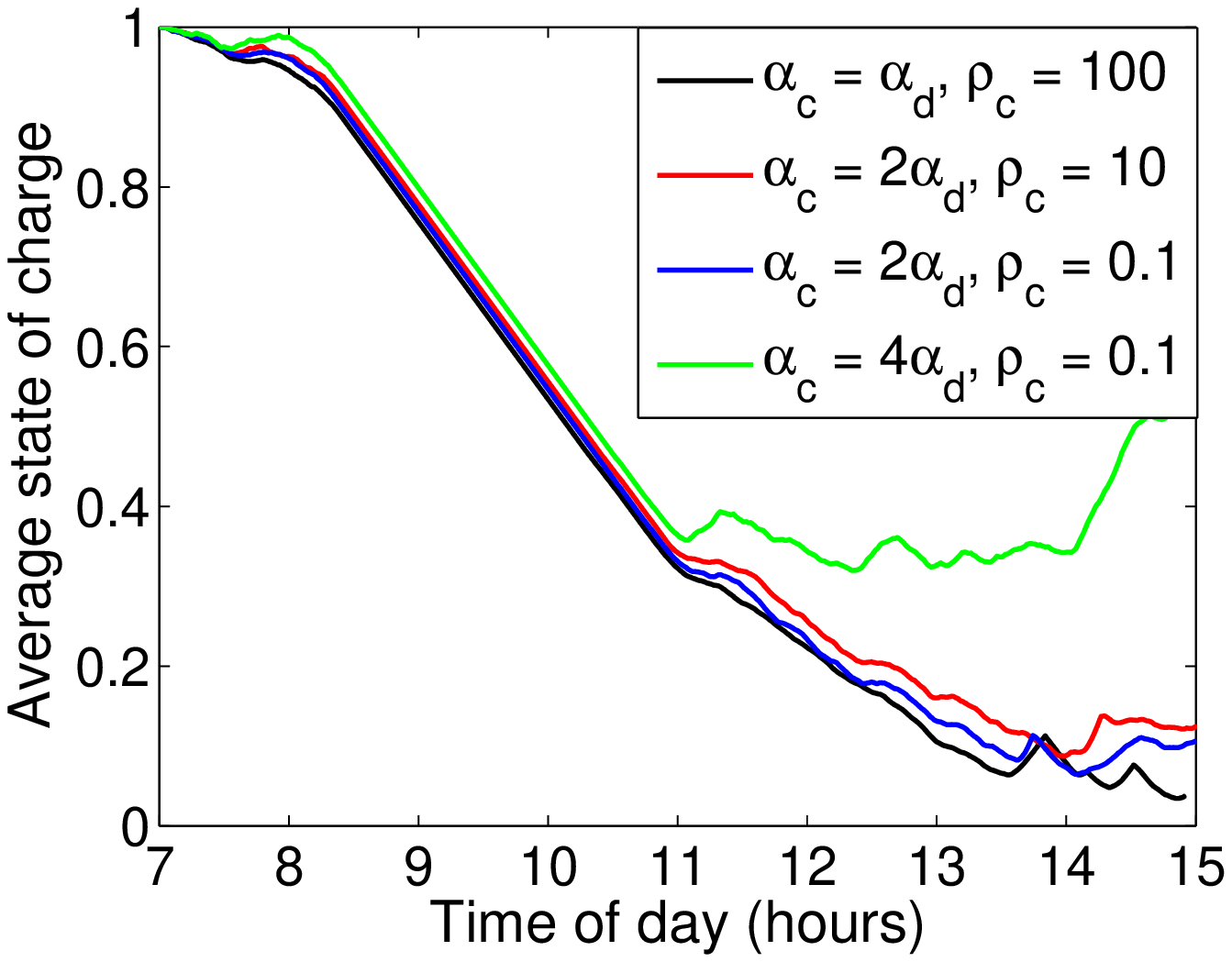}}
\vspace{-0.5em}
\caption{\ref{fig:chargeWaitTime}: Average customer wait time for April 2, with charging constraints. \ref{fig:chargeState}: Average vehicle charge as a function of time for different charging rates and final charging costs.}
\label{fig:charging}
\vspace{-1em}
\end{figure}

%\frmargin{
\section{Additional Model Extensions}\label{sec:extensions}
While we have mainly focused on extending the AMoD model for charging constraints, many other real-world constraints  can be directly incorporated into our modeling framework with little modification. We briefly touch on some of the possible extensions to highlight the flexibility and potential of our approach. 
\begin{enumerate}
\item \normalsize{\bf Limited number of charging stations}. Let $h_i, i \in \mathcal{N}$, represent the number of charging stations at each station. We can assume that $\sum_{i \in \mathcal{N}} h_i \geq |\mathcal{V}|$, that is, there is at least one charger for each vehicle in the system. Vehicles can only wait at a charging station, though they can still pick up and drop off passengers at stations with no free charging stations. This adds the constraint $\sum_{k \in \mathcal{V}} u^k_i(t) \leq h_i$ for all $i \in \mathcal{N}$. A limited number of charging stations will also promote rebalancing, as it forces vehicles to travel to stations with free charging stations (or likely a deficit of vehicles).

\item \normalsize{\bf Customer priorities}. Taking into account priority waiting involves simply adding a weighting matrix to the objective function. Rather than minimizing $\sum_{\tau = t}^{t+t_{\text{hor}}-1}\allowbreak J_x(x(\tau+1))$ we minimize $\sum_{\tau = t}^{t+t_{\text{hor}}-1} Q(\tau+1)J_x(x(\tau+1))$. In this way, we can give a higher priority to customers who have been waiting for a longer period of time, or assign the weights based on price incentives. This approach can also be used for customer arrivals with known time windows. 

% 
%\item \normalsize{\bf \frmargin{Modeling customer demand}{Done!}}. 
%Customer demand can usually be estimated from historical data based on features such as time-of-the-day, weekday/weekend, and weather conditions. Once a customer model is built, it can be fed into the MPC algorithm through $c_{ij}(t)$ to help the system anticipate future demand. The stability and disturbance rejection properties of the system with a customer demand model will be the subject of future research. If the AMoD system allows for advanced reservations, $c_{ij}(t)$ can help the system plan for trips reserved in advance. 

%\item \normalsize{\bf \frmargin{Heuristics for rebalancing}{Move to implementation}}. A specific distribution of vehicles across the stations that is well aligned with the customer demand will likely reduce wait time and improve the quality of service. At the same time, promoting a ``balanced'' vehicle distribution will decrease the time horizon needed for MPC to achieve stability, and hence reduce computational effort. To specify a vehicle distribution, we may wish to keep track of the number of vehicles at or enroute to station $i$ (this is used in \cite{MP-SLS-EF-DR:12,RZ-MP:15a}). Denote this quantity by $r^{\text{own}}_i(t)$, which can be propagated as a state variable as follows
%\vspace{-2mm}
%{\small\begin{equation}
%r^{\text{own}}_i(t+1) = r^{\text{own}}_i(t) + \sum_{k \in \mathcal{V}} \sum_{j \in \mathcal{N}} (v^k_{ji}(t) + w^k_{ji}(t))
%- \sum_{k \in \mathcal{V}} \sum_{j \in \mathcal{N}} (v^k_{ij}(t) + w^k_{ij}(t)). \notag
%\vspace{-3mm}
%\end{equation}}

\item \normalsize{\bf Interaction with the smart grid}. Electric vehicles may act as energy storage devices to enable intermittent renewable energy such as solar and wind \cite{WJM-CEBB-LDB:10}. Vehicles can ``sell'' their energy to the grid during peak hours and charge themselves during off-peak hours. If the charging/discharging schedule is known, the charging rates $\alpha_c$ can be adjusted accordingly to facilitate the energy transfer, at the same time maintaining quality of service in the AMoD system.
\end{enumerate}%}{We can summarize the entire thing in two lines in the conclusions. See the first action item in the conclusions below.}
%\vspace{-0.3em}
\section{Conclusions and Future Work} \label{sec:conc}
In this paper we presented a model predictive control approach to optimize vehicle scheduling and routing in an AMoD system. Our approach allows the easy integration of a number of real-world constraints, in particular we focused on charging constraints. We presented two MPC algorithms and rigorously showed that they are able to regulate an AMoD system (i.e., drive an initial customer demand to zero assuming no additional customers arrive over time). Algorithm performance for the case of dynamic arrivals was evaluated using real-world data through simulations. Overall, numerical results show that the proposed  MPC algorithms outperform previous control strategies for AMoD systems.

This paper leaves numerous important extensions open for
further research. %First, it is of interest to rigorously study the disturbance rejection properties of the proposed MPC algorithms.
%First, while this paper focuses on charging constraints, the expressivity of our model allows it to easily encode more complex constraints such as limited availability of charging stations, customer priorities (e.g. price-based pickup priorities) and even interaction with the smart grid \cite{WJM-CEBB-LDB:10}. We plan to further investigate this wider class of real-world constraints and to use our MPC algorithm to gain insight into their impact on system performance. 
First, we plan to address the algorithmic aspect of scaling up the MILP formulation of the %vehicle scheduling and routing  
optimization problem to large-scale (i.e., city-wide) systems. Since the computational complexity of the MILP formulation scales exponentially with the number of stations and vehicles, this will likely require the use of parallel architectures and ad hoc approximations. Second, we plan to study the inclusion of additional operational constraints (e.g., time windows), address the congestion aspect (currently roads are assumed to have infinite capacity), and study the impact on computation time.  Third, it is of interest to couple an AMoD system with alternative  mass transit options and develop an MPC approach for the optimal coordination algorithms of such
an intermodal system. 
Fourth, we plan to consider additional, larger-scale case studies to derive economic guidelines about the development of AMoD systems. Finally, we
plan to demonstrate the algorithms on real driverless vehicles
providing AMoD service in a gated community.

\section*{Acknowledgment}
The authors would like to acknowledge Luke Shimanuki and Leonardo Franco-Munoz for their contribution to the implementation of the algorithms in Section \ref{sec:results}.

\bibliographystyle{IEEEtran}
{\small
\bibliography{../../../bib/alias,../../../bib/main}
}

\ifarxiv

\ifimages
\begin{appendices}
\section{}
\else
\appendix
\fi
\label{apx:proofs}
In this section we provide the proofs of Proposition \ref{prop:feasible} and Theorems \ref{thm:mpc1} and \ref{thm:mpc2}. 

\begin{proof}[Proof of Proposition \ref{prop:feasible}]
Let $x^+ = [d^+_{ij}\;\, ^{T_i}p^{k+}_i\; u^{k+}_i]^\intercal$ and $u = [v^k_{ij}\; w^k_{ij}]^\intercal$. First, we note that if $v^k_{ij}$ satisfies \eqref{eq:constr4}, $d_{ij}^+ \in \{\mathbb{N} \cup {0}\}^{N^2}$ and is feasible. Next we show that $^{T_i}p^{k+}_i \in \{0,1\}^{\mathcal{D}}$ and satisfies \eqref{eq:pconstraint}. To show that $^{T_i}p^{k+}_i$ can only take on 0 or 1, first consider $T_i = T_{\text{max},i}$. In this case, by \eqref{eq:constr1}, $\sum_{j:t_{ji}-1 = T_{\text{max},i}} (v^k_{ji} + w^k_{ji}) \leq 1$, so $^{T_{\text{max},i}}p^{k+}_i \leq 1$. Now consider $T_i < T_{\text{max},i}$. Here, we only need to consider the case when $^{T_i+1}p^k_i = 1$. When this is the case, by \eqref{eq:construp}, $u^k_i = 0$ for all $i \in \mathcal{N}$. Also, by \eqref{eq:pconstraint}, $^0p^k_i = 0$ since $T_i + 1 > 0$ ($T_i \geq 0$). Putting these facts into \eqref{eq:uit}, we see that $\sum_{j \in \mathcal{N}} (v^k_{ij} + w^k_{ij}) = 0$ which proves that $^{T_i}p^{k+}_i \in \{0,1\}^{\mathcal{D}}$. To show constraint \eqref{eq:pconstraint} is satisfied, take the sum of \eqref{eq:tpki}: 
{\small\begin{equation}
\sum_i \sum_{T_i} \; ^{T_i}p^{k+}_i = \sum_{i} \sum_{T_i = 1}^{T_{\text{max},i}} \; ^{T_i}p^k_i + \sum_{i,j}(v^k_{ji} + w^k_{ji}).
\vspace{-2mm}
\label{eq:sump}
\end{equation}} 
By \eqref{eq:pconstraint} and \eqref{eq:constr1}, both terms on the right hand side are less than or equal to one. First consider when the second term is equal to one, from \eqref{eq:uit}, either $u^k_i = 1$ or $^0p^k_i = 1$, and in both cases, $\sum_{i} \sum_{T_i = 1}^{T_{\text{max},i}} \; ^{T_i}p^k_i = 0$. Now consider when the first term on the right hand side of \eqref{eq:sump} is equal to one. In this case, according to \eqref{eq:construp}, both $u^k_i$ and $^0p^k_i$ must be equal to zero, so by \eqref{eq:uit}, $\sum_{i,j}(v^k_{ji} + w^k_{ji}) = 0$. Hence $^{T_i}p^{k+}_i$ satisfies \eqref{eq:pconstraint}. 

To show that $u^{k+}_i$ satisfies \eqref{eq:construp}, take the sum of \eqref{eq:tpki} and \eqref{eq:uit}
{\small\begin{align*}
\sum_i u^{k+}_i + \sum_i \sum_{T_i} \; ^{T_i}p^{k+}_i 
&= \sum_i u^k_i + \sum_i\; ^0p^k_i - \sum_{i,j}(v^k_{ij} + w^k_{ij}) \\
&+ \sum_i \sum_{T_i = 1}^{T_{\text{max},i}}\; ^{T_i}p^k_i + \sum_{i,j}(v^k_{ji} + w^k_{ji}) \\
&= \sum_i u^k_i + \sum_i \sum_{T_i} \;^{T_i}p^{k}_i = 1.
\vspace{-2mm}
\end{align*}}
Hence, \eqref{eq:construp} is satisfied and the proposition is proven.
\end{proof}

The notions of N-step and $\infty$-step reachable sets for the undisturbed system \eqref{eq:linear} are used when proving Lyapunov stability of our MPC algorithms in Theorems \ref{thm:mpc1} and \ref{thm:mpc2}.

%We can now define the N-step and $\infty$-step reachable sets for the undisturbed system \eqref{eq:linear}. These definitions will be useful when proving Lyapunov stability of our MPC algorithms.
\begin{definition}[N-step Reachable Set]
Given an initial condition $x(0) \in \mathcal{X}$, the N-step reachable set is defined recursively as

\vspace{-2mm}
{\small 
\begin{align}
\mathcal{R}_{i+1} := \left\{x^+ \in \mathcal{X} \mid \exists x \in \mathcal{R}_i, u \in \mathcal{U}(i) \emph{ s.t. } x^+ = Ax + Bu\right\},
\end{align}
}
for $i = 0...N-1$ and $\mathcal{R}_0 = x(0)$. 
\end{definition}
\begin{definition}[$\infty$-step Reachable Set]
Given $x(0) \in \mathcal{X}$, the $\infty$-step reachable set of system \eqref{eq:linear} subject to \eqref{eq:uset} is
\begin{align}
\mathcal{R}_{\infty} := \limsup_{N \rightarrow \infty} \mathcal{R}_N,
\end{align}
where the above limit is in a set-theoretical sense (i.e., $\limsup_{N \rightarrow \infty} \mathcal{R}_N = \cap_{N\geq 1} \cup_{m\geq N} \mathcal{R}_m$).
\end{definition}

Our proof of Theorem  \ref{thm:mpc1} relies on the following theorem for set-valued Lyapunov functions, the proof of which can be found in \cite{LM:05}, Theorem 4. 
\renewcommand{\thetheorem}{\Alph{section}.\arabic{theorem}}

\begin{theorem}[Lyapunov stability for set-valued functions]\label{thm:1}
Let $\mathcal{R}$ be a finite dimensional Euclidean space and consider a continuous map $f : \mathbb{N} \times \mathcal{R} \rightarrow \mathcal{R}$ giving rise to the discrete-time system
\begin{equation}
x(t+1) = f(t, x(t)).
\label{eq:discrete}
\end{equation}
Let $\Xi$ be the collection of equilibrium solutions of \eqref{eq:discrete} and $\mathcal{X}^e$ be the set of equilibrium points corresponding to $\Xi$. Let $W : \mathcal{R} \rightrightarrows \mathcal{R}$ be an upper semi-continuous set-valued Lyapunov function satisfying
\begin{enumerate}
\item $x \in W(x)$ for all $x \in \mathcal{R}$, 
\item $W(x^e) = \{x^e\}$ for all $x^e \in \mathcal{X}^e$,
\item $W(x(t+1)) \subseteq W(x(t))$ for all $x(t) \in \mathcal{R}$.
\end{enumerate}
Then, system \eqref{eq:discrete} is uniformly stable with respect to $\Xi$ in the sense of Lyapunov. 
%If $W(x)$ is bounded for all $x \in \mathcal{R}_\infty$, then system \eqref{eq:linear} is uniformly bounded with respect to $\Xi$. 
If additionally,
\begin{enumerate}
\setcounter{enumi}{3}
\item there exists a function $\mu : \text{Im}(W) \rightarrow \mathbb{R}_{\geq 0}$, bounded on bounded sets, such that
\begin{equation}
\mu\left( W(x(t+1))\right) < \mu\left(W(x(t))\right)
\label{eq:mu}
\end{equation}
\end{enumerate}
for all $x(t) \in \mathcal{R} \setminus \mathcal{X}^e$, then $x(t) \rightarrow x^e \in \mathcal{X}^e$ as $t \rightarrow \infty$ and the system is asymptotically stable with respect to $\Xi$.
\end{theorem}

We are now in a position to present the proof of Theorem \ref{thm:mpc1}.

\begin{proof}[Proof of Theorem \ref{thm:mpc1}]

To show that Algorithm \ref{alg:mpc1} solves the ARP, we need only to look at the subspace $\mathcal{Y}$ of $\mathcal{X}$, where
\begin{equation}
\mathcal{Y} = \left\{\tilde{x} = d_{ij} \mid d_{ij} \in (\mathbb{N} \cup \{0\})^{N^2}, \quad d_{ii} = 0 \right\}.
\end{equation}
This is because the objective $J_x(x(t))$ considered in the ARP is only a function of $d_{ij}$'s. The dynamics of $\tilde{x}$ (referred to as the reduced state) follow \eqref{eq:dij}, and the corresponding reduced control becomes $\tilde{u}(t) = v^k_{ij}(t)$ (note that all other constraints in Section \ref{sec:linear} must still be satisfied). Equation \eqref{eq:dij} can then be written as 
\begin{equation}
\tilde{x}(t+1) = \tilde{x}(t) + \tilde{B}\tilde{u}(t),
\label{eq:linearr}
\end{equation}
where $\tilde{B}$ is the matrix realization of the summation in \eqref{eq:dij}. We can define the reduced N-step reachable set $\mathcal{Y}_N$ and the $\infty$-step reachable set $\mathcal{Y}_\infty$ as the appropriate subspaces of $\mathcal{R}_N$ and $\mathcal{R}_\infty$, respectively. For $t_{\text{hor}}$ transitions, we can write
\begin{align}
\tilde{x}(t+t_{\text{hor}}) &= \tilde{x}(t+t_{\text{hor}}-1) + \tilde{B}\tilde{u}(t+t_{\text{hor}}-1) \notag \\
&= \tilde{x}(t) + \tilde{B}(\tilde{u}(t) + \ldots + \tilde{u}(t+t_{\text{hor}}-1))
\label{eq:thor}
\end{align}
We can further rescale the time variable so that one new time step (denoted by $T$) is equivalent to $t_{\text{hor}}$ old time steps. With this, we can rewrite \eqref{eq:thor} as
\begin{equation}
\tilde{x}(T+1) = \tilde{x}(T) + \tilde{B} \tilde{U}(T)
\label{eq:rescaled}
\end{equation} 
where $\tilde{U}(T) = \tilde{u}(t) + \ldots + \tilde{u}(t+t_{\text{hor}}-1)$.

Now, using our reduced system written in the form of \eqref{eq:rescaled}, consider the following set-valued Lyapunov function candidate
\begin{equation}
W(\tilde{x}) := \left\{\tilde{x} + \tilde{B}\tilde{U} \in \mathcal{Y}_\infty \mid J_x(\tilde{x}+\tilde{B}\tilde{U}) \leq J_x(\tilde{x})  \right\}.
\end{equation}
First, note that the reduced system \eqref{eq:rescaled} is persistently feasible, since $\tilde{u}(t) = v^k_{ij}(t) = 0$ is always a feasible control input (this is the case where no customers are serviced). Next, we show that $W(\tilde{x})$ is upper semi-continuous. It is necessary and sufficient that the graph of $W(\tilde{x})$ be a closed set \cite[p.42]{JA-HF:90}. The graph of $W$, namely
{\small \begin{equation}
\text{graph}(W) := \left\{ (\tilde{x}, \tilde{x}+\tilde{B}\tilde{U}) \mid \tilde{x} \in \mathcal{Y}_\infty, J_x(\tilde{x}+\tilde{B}\tilde{U}) \leq J_x(\tilde{x}) \right\},
\end{equation}
}
is closed because the state space is finite, hence $W(\tilde{x})$ is upper semi-continuous.

The equilibrium point we wish to converge to corresponds to $\arg\min J_x(x(t))$ which is $\tilde{x} = \mathbf{0}$ (the state where all customers have been served). Hence, $\mathcal{X}^e = \{\mathbf{0}\}$. 
To satisfy the first condition of Theorem \ref{thm:1}, set $v^k_{ij}(\tau) = 0$ for $\tau = t,\ldots,t+t_{\text{hor}}-1$. This is the same as setting $\tilde{U} = \mathbf{0}$ and hence $\tilde{x} \in W(\tilde{x})$. For the second condition, since $x^e = \mathbf{0}$ and $J_x(x^e) = 0$, we have $W(x^e) = \{x^e\}$. Note that because of constraint \eqref{eq:constr4}, $\tilde{U} = \mathbf{0}$ when $\tilde{x} = \mathbf{0}$. 

To show that $W(\tilde{x})$ satisfies the third condition, let $z \in W(\tilde{x}(T+1))$. By definition, there exists a sequence of feasible inputs $\tilde{V}$ such that $z = \tilde{x}(T+1) + \tilde{B}\tilde{V}$ and $J_x(z) \leq J_x(\tilde{x}(T+1))$. By \eqref{eq:rescaled}, there exists $\tilde{U}$ such that $\tilde{x}(T+1) = \tilde{x}(T) + \tilde{B}\tilde{U}$. Hence, there is a feasible sequence of inputs $\tilde{U}+\tilde{V}$ such that $z = \tilde{x}(T) + \tilde{B}(\tilde{U}+\tilde{V})$, and since $J_x(z) \leq J_x(\tilde{x}(T+1)) \leq J_x(\tilde{x}(T))$, we have $W(\tilde{x}(T+1)) \subseteq W(\tilde{x}(T))$. 

Finally, let
\begin{equation}
\vspace{-1mm}
\mu(W(\tilde{x}(T))) = \min_{z \in W(\tilde{x}(T))} J_x(z).
\vspace{-1mm}
\end{equation}
Clearly, $\mu(W(\tilde{x}))$ is bounded. To show that $\mu(W(\tilde{x}(T+1))) < \mu(W(\tilde{x}(T)))$ is equivalent to showing that the number of waiting customers decreases from $T$ to $T+1$ under a sequence of feasible control actions given by the solution of Algorithm \ref{alg:mpc1}. According to \eqref{eq:dij}, $J_x$ is minimized when $\sum_{i,j \in \mathcal{N}, k \in \mathcal{V}} v^k_{ij}$ is maximized, and $\mu(W(\tilde{x}))$ will decrease as long as at least one value in $v^k_{ij}(t),\ldots, v^k_{ij}(t+t_{\text{hor}}-1)$ is nonzero. Thus, \eqref{eq:mu} will be satisfied if we find an upper bound on $t_{\text{hor}}$ that will guarantee at least one of $v^k_{ij}(t),\ldots, v^k_{ij}(t+t_{\text{hor}}-1)$ will be nonzero. The variables $v^k_{ij}(t)$ and $w^k_{ij}(t)$ are governed by \eqref{eq:uit}, which states that either $v^k_{ij}(t)$ or $w^k_{ij}(t)$ can only be nonzero if vehicle $k$ has been waiting or has just arrived at a station. The time it takes for a vehicle to arrive at a station is upper bounded by $\max_{i,j} t_{ij}$. However, the station that the vehicle arrives at may not have customers waiting. In this case, $w^k_{ij}$ is needed to send the vehicle to a station with customers, which takes an amount of time upper bounded by $\max_{i,j} t_{ij}$. Thus, if $t_{\text{hor}} \geq 2 \max_{i,j} t_{ij}$, \eqref{eq:mu} is satisfied for all $\tilde{x}(t) \neq \mathbf{0}$. Hence, by Theorem \ref{thm:1}, the MPC algorithm is asymptotically stable and $\tilde{x}(t) \rightarrow \mathbf{0}$ as $t \rightarrow \infty$, which completes the proof.
\end{proof}

\begin{proof}[Proof of Theorem \ref{thm:mpc2}]
The proof for this theorem follows the same procedure as the proof of Theorem \ref{thm:mpc1}. The main difference is the choice of the time horizon. Consider again the case where vehicle $k$ is enroute from station $j$ to station $i$. Suppose there are no customers at station $i$, so once the vehicle arrives, it must travel back to station $j$ for the next pickup. However, once the vehicle arrives at station $i$, it is completely depleted of charge and needs to charge for $\frac{\alpha_d}{\alpha_c} t_{ij}$ amount of time before departing. Once the vehicle arrives back at station $j$, it is again depleted of charge and must charge again before servicing the customer. The total time this takes consists of two traveling periods and two charging periods, with the travel time upper bounded by $\max_{i,j \in \mathcal{N}} t_{ij}$ and charging time upper bounded by $\frac{\alpha_d}{\alpha_c} \max_{i,j \in \mathcal{N}} t_{ij}$. This completes the proof.
\end{proof}

\ifimages
\section{}
This section contains the simulation results from March 19, April 9, and April 30. 
\begin{figure}[h]
\centering
\includegraphics[width=.49\textwidth]{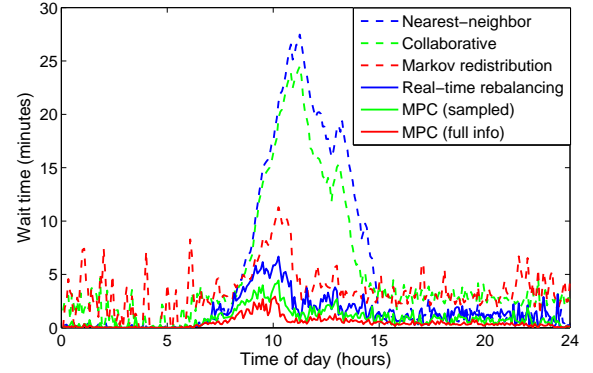} 
\caption{Average customer wait times throughout the day (March 19, 2012) for all dispatch algorithms.}
\label{fig:amodWaitTimeMar19}
\end{figure}
\begin{figure}[h]
\centering
\includegraphics[width=.49\textwidth]{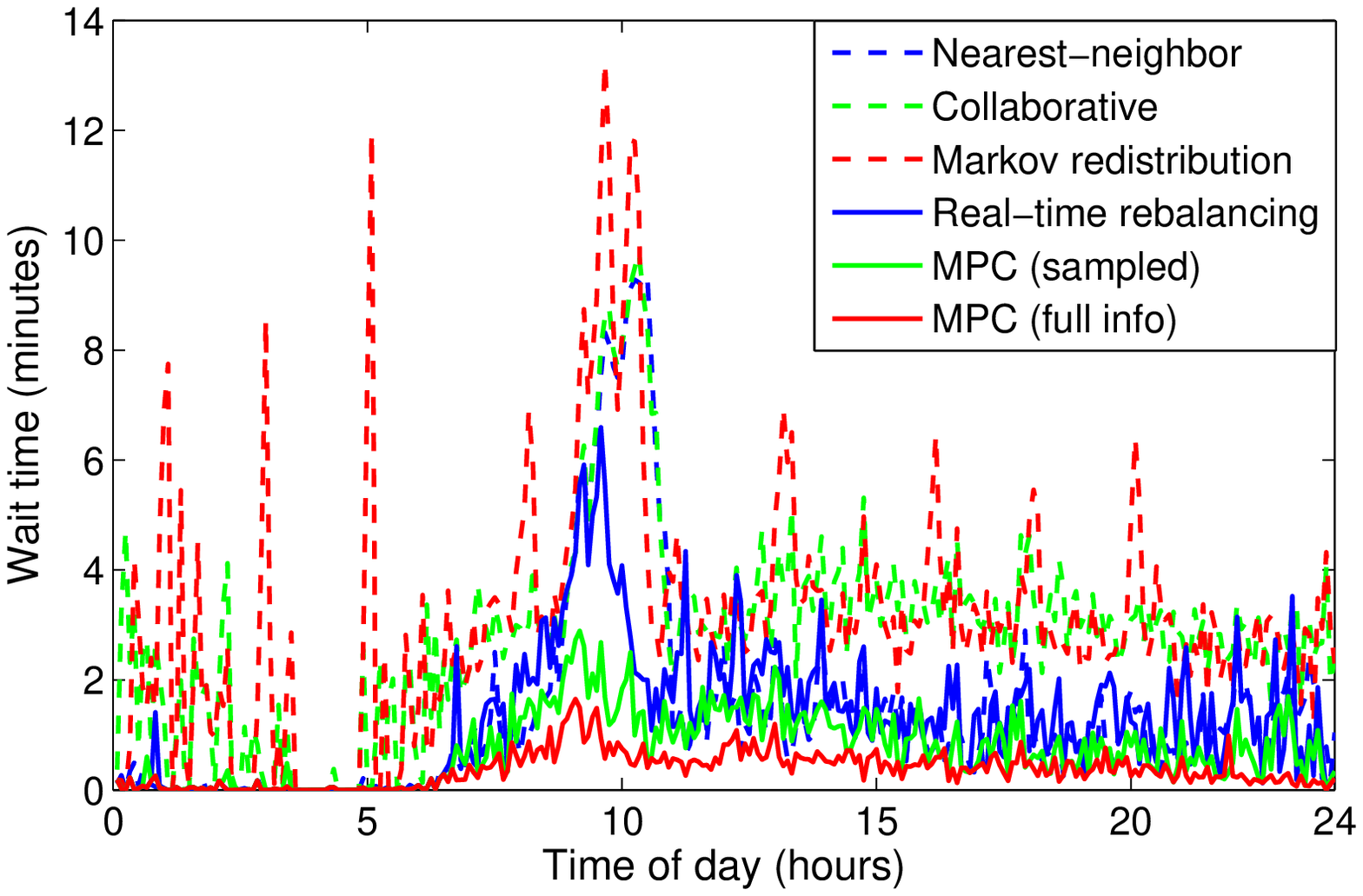} 
\caption{Average customer wait times throughout the day (April 9, 2012) for all dispatch algorithms.}
\label{fig:amodWaitTimeApr9}
\end{figure}
\begin{figure}[h]
\centering
\includegraphics[width=.49\textwidth]{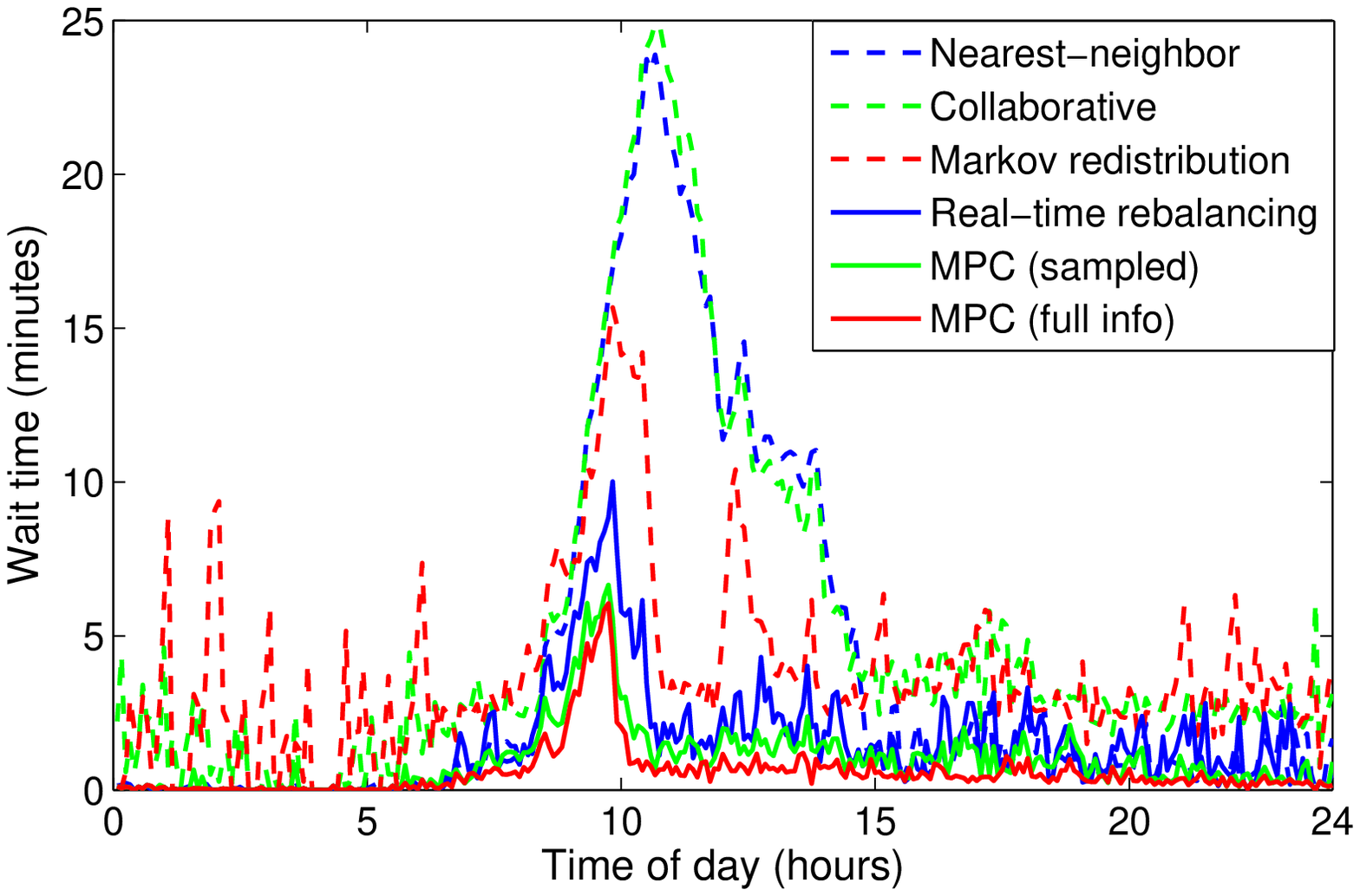} 
\caption{Average customer wait times throughout the day (April 30, 2012) for all dispatch algorithms.}
\label{fig:amodWaitTimeApr30}
\end{figure}

\fi

\ifimages
\end{appendices}
\fi

\fi

\end{document}